\documentclass[trans_jour]{IEEEtran}
\usepackage[T1]{fontenc}
\usepackage{algorithm}
\usepackage{algpseudocode}%
\usepackage{enumitem}
\usepackage{graphicx}
\usepackage{url}
\ifCLASSINFOpdf
\else
\fi
\usepackage{amsmath}
\usepackage{amsthm}
		\newtheorem{theorem}{Theorem}
		\newtheorem{lemma}{Lemma}
    \newtheorem{proposition}{Proposition}
		\newtheorem{problem}{Problem}
		\newtheorem{definition}{Definition}
		\newtheorem{fact}{Fact}

		\newtheorem{remark}{Remark}

\usepackage[cmintegrals]{newtxmath}
\begin{document}
%
\title{Robust Cell-Load Learning with a Small Sample Set\\}
\author{Daniyal~Amir~Awan,~\IEEEmembership{Student Member,~IEEE,}
        Renato~L.G.~Cavalcante,~\IEEEmembership{Member,~IEEE,}
        and~Slawomir~Stanczak,~\IEEEmembership{Senior Member,~IEEE.}
				
}
%
%

\maketitle

\begin{abstract}
Learning of the cell-load in radio access networks (RANs) has to be performed within a short time period. Therefore, we propose a learning framework that is robust against uncertainties resulting from the need for learning based on a relatively small training sample set. To this end, we incorporate prior knowledge about the cell-load in the learning framework. For example, an inherent property of the cell-load is that it is monotonic in downlink (data) rates. To obtain additional prior knowledge we first study the feasible rate region, i.e., the set of all vectors of user rates that can be supported by the network. We prove that the feasible rate region is compact. Moreover, we show the existence of a Lipschitz function that maps feasible rate vectors to cell-load vectors. With these results in hand, we present a learning technique that guarantees a minimum approximation error in the worst-case scenario by using prior knowledge and a small training sample set. Simulations in the network simulator NS$3$ demonstrate that the proposed method exhibits better robustness and accuracy than standard multivariate learning techniques, especially for small training sample sets.  
\end{abstract}

\IEEEpeerreviewmaketitle

\section{Introduction}\label{sec:introduction}
The fifth-generation (5G) networks will be based on orthogonal frequency-division multiple access (OFDMA). Due to inter-cell interference, radio resource management (RRM) and performance optimization in these networks are challenging. In fact, many RRM problems in OFDMA-based networks, such as small-scale optimal assignment of time-frequency resource blocks and powers to users, have been shown to be NP-hard \cite{Wong04}. Recent research has therefore focused on the development of frameworks that capture the essence of OFDMA-based networks, while leading to a tractable problem formulation. An example of such a framework is the \textit{non-linear load-coupling model} proposed in \cite{Siomina2012,Fehske2012,Majewski2010}. In this framework the \textit{cell-load} at a base station is the fraction of time-frequency resource blocks that are used to support downlink data rates (henceforth simply rates). With this model, and given some power budget that can be used for transmission, one can estimate the cell-load required at each base station to support given rates. 

  The study in \cite{Ho2014} shows the intuitive result that the cell-load is monotonic in rates. The interference coupling between cells implies that increasing the rates in an arbitrary cell increases the cell-load at each base station, which also increases the inter-cell interference.\footnote{For brevity, we assume that cells are not mutually orthogonal.} So, it is important for a base station to have a reliable forecast of the cell-load before serving higher rate demands from its associated users. Therefore, cell-load learning can be used to make radio resource management and self-organizing-network (SON) algorithms more reliable and efficient. 

	Cell-load learning is also a vital part of energy saving mechanisms in radio access networks (RANs). For instance in \cite{Skillermark2012}, the value of the cell-load is used as an input to a simple heuristic algorithm that switches off base station antennas when the cell-load is low. Large gains in energy savings are reported with minimal effect on the cell sum throughput. The same concept can be used in the case of virtual base station formations in \textit{cloud RANs} \cite{XWang2016}. In these virtual systems some power-hungry components of a RAN (digital signal processors, line cards, fronthaul, etc.) are virtualized in a central location, and these components can be allocated on-demand to cells according to the cell-load. Therefore, given RAN data traffic (or rates) predictions, the corresponding cell-load forecasts can enable us to proactively manage network components for energy savings. 
%

\subsection{The Need for Robust Cell-Load Learning}\label{sec:why_robust_cell_load_estimation}
Note that even though the load-coupling model has been shown to work sufficiently well in predicting the cell-load in some scenarios \cite{Fehske2012, Mogensen2007, Shen2015T}, models are only idealizations and in general they do not capture all the intricacies of dynamic wireless environments. Therefore, our objective is to directly learn the underlying function that maps user rates to cell-load values given a training sample set consisting of rate vectors and the corresponding measured cell-load vectors. To improve the learning process, we use the load-coupling model to study some salient aspects of the relationship between rates and the cell-load. We use these aspects as prior knowledge in the learning process. 

  Compared to the core network, the RAN data traffic is volatile and it shows irregular patterns throughout a day because of the unpredictable nature of user activity and relatively fast changes in the network topology \cite{greennets}. Therefore, the underlying statistics (i.e., the joint probability distribution) of rates and the corresponding cell-load values, which are part of the so-called \textit{environment}, can be assumed to remain constant for only a short time. This implies that a training sample set must be acquired during this short time before the environment changes, since otherwise the sample set can be rendered useless for predicting future cell-load values. However, in general, the smaller the sample set, the larger the uncertainty about the underlying phenomenon, which makes large prediction errors on unseen rates more probable. 
	
	In uncertain situations we need ``robust'' learning methods that provide a guaranteed worst-case performance under uncertainty. The objective of this study is to develop such a robust learning framework. Our method is optimal in the sense that it minimizes the worst-case or maximum error of approximation which is a classical robust optimization problem (see, e.g., \cite{Sukharev1992,traub1980,Weinberger1959,Calliess2014}). This means that no matter how small the training sample set is, we are guaranteed the best worst-case error. Our method involves only low-complexity and stable mathematical operations and its theoretical properties are very well understood. The above mentioned optimization problem is solved by explicitly incorporating prior knowledge regarding the Lipschitz continuity of the function to be approximated. By incorporating additional prior knowledge concerning monotonicity of the function, we further reduce the worst-case error. 

   We point out that our framework is different to many modern conventional machine learning frameworks that target mean or average performance rather than the worst-case performance we consider in this study. The performance of many current complex learning methods, such as deep neural networks (DNNs), is often dependent on the availability of a large training (or pre-training) sample set. 
Including prior knowledge in these frameworks to reduce the reliance on large training sets is not easy, and it is often discouraged \cite{Marcus2018}. Even if some prior knowledge could be enforced in neural networks (as in \cite{Diligenti}), it is theoretically unclear whether (or how) this enables neural networks to learn better. This makes DNNs ill-suited to our setting because we consider learning with very small training sample sets.
%

\subsection{Related Work}\label{sec:some_related_work}
The load-coupling model \cite{Siomina2012,Fehske2012,Majewski2010} is commonly used when designing networks according to the long-term evolution (LTE) standard. Recently it has also attracted attention in the context of 5G networks \cite{You2018}. More specifically, the load-coupling model has been used in various optimization frameworks dealing with different aspects of network design including data offloading \cite{Ho2014}, proportional fairness \cite{miguel2016}, energy optimization \cite{daniyal2016,Pollakis2016,Ren2014}, and load balancing \cite{siomina2012b}. In the context of energy savings, and by using the theory of \textit{implicit functions} \cite{Krantz2003}, the study in \cite{Ren2014} shows that there exists a continuously differentiable function relating user associations with the base stations to the cell-load. In contrast to \cite{Ren2014}, the user association is assumed to be fixed in this study; we study the relationship between downlink rates and the cell-load and we incorporate this prior knowledge in our learning framework. Previous studies dealing with cell-load estimation, for instance, in the context of data offloading \cite{Ho2014} and maximizing the scaling-up factor of traffic demand \cite{siomina2015}, have used load coupling model driven methods that require information about channel gains, powers, etc.. Most of these methods employ iterative algorithms to estimate the cell-load for given downlink rates and other parameters by exploiting the fact that the cell-load is the fixed point of the \textit{standard interference mapping} \cite{yates95} that is constructed using the network information. In contrast, we directly learn the underlying function that maps feasible rates, i.e., downlink data rates that can be supported by the network, to the observed cell-load in the network using a sample training set and prior knowledge. Our framework, therefore, does not require information about powers, channels, etc.. 
	
	 Inclusion of prior knowledge in the form of constraints, known properties, and logic has also been widely used in other areas, such as optimal control \cite{Sun2019,Qiu2019}, to deal with uncertainty. However, incorporating prior knowledge in machine learning algorithms for multivariate data\footnote{Multivariate data in this context means that the input argument (or domain) of the function to be approximated has an arbitrary dimension.} with arbitrary dimensions is difficult, and most of the well-known algorithms either do not preserve the ``shape'' (i.e., known properties such as monotonicity, continuity, etc.) of the underlying function or they become too complex for high-dimensional data \cite{Beliakov2005}. An inherent property of the cell-load is that it is monotonic in rates. The study in \cite{Kotlowski16} shows that monotonicity is difficult to incorporate in popular online learning methods even in the case of univariate data. In \cite{Beliakov2005} the author proposes a shape preserving multivariate approximation of scalar monotonic functions that are also \textit{Lipschitz}. The author shows that Lipschitz continuity of the function to be approximated allows for computing tight upper and lower bounds on the function values. Using these bounds one can obtain an optimal solution in the sense that this solution minimizes a worst-case error of approximation \cite{Sukharev1992,traub1980,Weinberger1959}. Furthermore, the approximation preserves both the monotonicity and the Lipschitz continuity of the underlying function.
%

\subsection{Our Contribution}\label{sec:contribution}
This study deals with the problem of learning cell-load in RANs as a function of downlink rates given a relatively small training sample set. The assumption of small training sample sets is crucial because modern RAN networks do not permit a long observation and sample acquisition period (see Section \ref{sec:why_robust_cell_load_estimation}). 
To cope with this limitation, we propose a robust learning framework that guarantees a minimum worst-case error of approximation. To achieve robustness, we incorporate prior knowledge about the cell-load and its relationship with rates. We show that the incorporation of prior knowledge enables us to provide explicit tight bounds that cannot be achieved by using a sample set alone, no matter how large the sample set is. 

   In the following we summarize the main contributions of this study.
\begin{enumerate}
\item We study the feasible rate region which is defined as the set of all rates that can be supported by the network. In the conference version of this study \cite{AwanIcassp18} we stated without proof that the feasible rate region is compact. In this work we provide a formal proof for this assertion along with some other related results. 
\item In particular, we show that there exists a function that maps rates to the cell-load and that this function is monotonic and Lipschitz continuous over the feasible rate region.  
\item We use the prior knowledge developed in 1) and 2) to perform robust learning of the cell-load by using the framework of \textit{minimax approximation} \cite{Sukharev1992,traub1980,Weinberger1959}. Note that, this technique cannot be directly used without the prior knowledge above.
\item In contrast to \cite{Beliakov2005}, where the main concern is to preserve the monotonicity, we show theoretically and by experiments that including the prior knowledge regarding monotonicity results in reduced uncertainty. 
\item Our machine learning framework does not require network information such as powers and channel gains in contrast to traditional cell-load approximation methods. 
The guaranteed performance of our framework with small sample sets makes it suitable in such scenarios where other learning frameworks such as DNNs cannot be applied. 
\item In contrast to the conference version, we perform simulations in the network simulator NS$3$ to demonstrate the performance of the algorithm in a realistic cellular wireless network. We compare our framework with standard multivariate learning techniques and show that our method outperforms these standard techniques for small sample sizes.   
\end{enumerate}

\subsection{Overview}
The remainder of this study is organized as follows. Section~\ref{sec:maths} provides the mathematical background and results that are used throughout the study. 
Section~\ref{sec:system_model} presents the non-linear load coupling model. In Section~\ref{sec:feasible_rate_demand_set} we provide our results on the feasible rate region. In Section \ref{section:the_learning_problem} we discuss the robust optimization problem for cell-load learning along with some more related results. Section~\ref{sec:cell_load_approximation} deals with the implementation of the cell-load learning framework developed in this study in a wireless network. Finally, in Section~\ref{sec:numerical_evaluation}, empirical analysis is performed by simulations in the network simulator (NS$3$).

\section{Mathematical Background}\label{sec:maths}
Throughout this study $\mathbb{R}$, $\mathbb{R}_{\geq 0}$, and $\mathbb{R}_{>0}$ denote the sets of reals, non-negative reals, and positive reals, respectively. We denote by $\|\cdot\|$ and $\|\cdot\|_{\infty}$ the usual Euclidean norm and $l_{\infty}$ norm in $\mathbb{R}^m$, respectively. The sets of non-negative integers and natural numbers are denoted by $\mathbb{Z}_{\geq 0}$ and $\mathbb{N}:=\mathbb{Z}_{\geq 0}\setminus\{0\}$, respectively. We define $\overline{N_{1},N_{2}}:=\left\{N_1,N_{1}+1,N_{1}+2,\ldots,N_2\right\}$, $N_1, N_2 \in \mathbb{Z}_{\geq 0}$ with $N_1\leq N_2$. We denote by $(\mathbf{x})_{+}$ the operation $\max\left\{\mathbf{x},\mathbf{0}\right\}$ for a vector $\mathbf{x} \in \mathbb{R}^{N}$, where the $\max$ is taken component-wise and $\mathbf{0}$ is the all-zero vector. For two vectors $\mathbf{x}$ and $\mathbf{y}$, the inequality $\mathbf{x}\leq\mathbf{y}$ should be understood component-wise.

    Let $\mathcal{S}$ be a normed vector space equipped with a norm $\|\cdot\|_{\mathcal{S}}$ and its induced metric $d_{\mathcal{S}}:\mathcal{S}\times\mathcal{S}\to \mathbb{R}_{\geq 0}:(\mathbf{s}_{o},\mathbf{s}) \mapsto \|\mathbf{s}_{o}-\mathbf{s}\|_{\mathcal{S}}$. We denote by $\mathcal{B}_{\mathcal{S}}(\mathbf{s}_{o},\delta):=\{\mathbf{s}\in \mathcal{S}| \|\mathbf{s}-\mathbf{s}_{o}\|_{\mathcal{S}}<\delta\}$ the open-ball of radius $\delta>0$ centered at $\mathbf{s}_{o}\in \mathcal{S}$. A sequence $(\mathbf{s}_n)_{n \in \mathbb{N}} \subset \mathcal{S}$ is said to converge (in norm) to $\mathbf{s} \in \mathcal{S}$ if $\|\mathbf{s}_n-\mathbf{s}\|_{\mathcal{S}} \to 0$ \cite[Page~26]{Luenberger1997}. 

We now define the concepts of \textit{boundedness}, \textit{closedness}, and \textit{compactness} that we use throughout this study.
\begin{definition}[\textit{Boundedness, Closedness, and Compactness}]\cite[Chapter~2]{Luenberger1997}
Consider a set $\mathcal{K}$ in the normed space $(\mathcal{S},\|\cdot\|_{\mathcal{S}})$.
\begin{itemize}
\item[a).] Boundedness: $\mathcal{K}$ is bounded if $(\exists L \geq 0)$ $(\forall \mathbf{k} \in \mathcal{K})$ $\|\mathbf{k}\|_{\mathcal{S}}\leq L$. 
\item[b).] Closedness: $\mathcal{K}$ is closed if and only if every convergent sequence $(\mathbf{k}_n)_{n \in \mathbb{N}} \subset \mathcal{K}$ has a limit in $\mathcal{K}$. 
\item[c).] Compactness: $\mathcal{K}$ is compact if every sequence $(\mathbf{k}_n)_{n \in \mathbb{N}} \subset \mathcal{K}$ has a convergent subsequence with a limit in $\mathcal{K}$.   
\end{itemize}
\label{def:boundedness_closedness_compactness}
\end{definition}
   In this study we consider the space $C(\mathcal{X},\mathcal{Y})$ of vector-valued continuous functions mapping $\mathcal{X} \subset \mathbb{R}_{>0}^{N}$ to $\mathcal{Y} \subset \mathbb{R}_{\geq 0}^M$. For a function $\mathbf{g} \in C(\mathcal{X},\mathcal{Y})$ its $i$th component $(i \in \overline{1,M})$ $g_i:\mathcal{X}\rightarrow\mathbb{R}_{\geq 0}$ is a scalar continuous function. We equip $C(\mathcal{X},\mathcal{Y})$ with the uniform norm \cite[Page~23]{Luenberger1997}
\begin{equation}
\left\|\mathbf{g}\right\|_{C(\mathcal{X})}=\sup_{\mathbf{x}\in\mathcal{X}}\max_{1 \leq i \leq M}g_{i}(\mathbf{x}).
\label{eqn:cheb_norm}
\end{equation}
   If $\mathcal{X}$ is compact, then the $\sup$ is attained according to the \textit{extreme value theorem} \cite{Munkres2000} because the $\max$ operation\footnote{The usage of $\max$ in \eqref{eqn:cheb_norm} is different to the component-wise $\max$ in $\max\{\mathbf{x},\mathbf{0}\}$. The distinction between the two usages shall be clear by the context in which they are used.} preserves continuity. 
 
   We now present some important concepts to keep the study as self-contained as possible. These concepts are essential to understanding our results in Section \ref{sec:feasible_rate_demand_set} and in Section \ref{section:the_learning_problem}.
 \begin{definition}[\textit{Monotonic Function}]
Let $\mathcal{X} \subset \mathbb{R}^{N}_{>0}$ and  $\mathcal{Y} \subset \mathbb{R}^{M}_{\geq 0}$. A function $\mathbf{f}:\mathcal{X}\rightarrow \mathcal{Y}$ is said to be monotonic if $(\forall \mathbf{x} \in \mathcal{X})$ $(\forall \mathbf{y} \in \mathcal{X})$ $\mathbf{x}\leq\mathbf{y}\Rightarrow\mathbf{f}(\mathbf{x})\leq\mathbf{f}(\mathbf{y})$.
\label{def:mf}
\end{definition}
 \begin{definition}[\textit{$\mathbf{L}$-Lipschitz function}]
Consider $\mathbf{f}\subset C(\mathcal{X},\mathcal{Y})$ and a vector $\mathbf{L}:=[L_1,L_2,\cdots,L_M]^{\intercal}\in \mathbb{R}_{\geq 0}^M$. We say that $\mathbf{f}$ is $\mathbf{L}$-Lipschitz on $\mathcal{X}$ if $(\forall i\in \overline{1,M})$ $(\forall \mathbf{x} \in \mathcal{X})(\forall \mathbf{y} \in \mathcal{X})\left|f_{i}(\mathbf{x})-f_{i}(\mathbf{y})\right|\leq L_{i}\left\|\mathbf{x}-\mathbf{y}\right\|$. 
\label{def:lf}
\end{definition}
\begin{definition}[\textit{$\mathbf{L}$-Lipschitz-Monotonic Function}]
We say that $\mathbf{f}\subset C(\mathcal{X},\mathcal{Y})$ belongs to the class of \textit{$\mathbf{L}$-Lipschitz-Monotonic Functions} (LIMF) if $\mathbf{f}$ is monotonic and there exists $\mathbf{L} \in \mathbb{R}_{\geq 0}^M$ such that $\mathbf{f}$ is $\mathbf{L}$-Lipschitz. 
\label{def:lmf}
\end{definition}
Note that a function $\mathbf{f}\in C(\mathcal{X},\mathcal{Y})$ is continuous at $\mathbf{x}_{o} \in \mathcal{X}$ if given $\epsilon > 0$, there exists $\delta_{\mathbf{x}_{o}}>0$ such that $(\forall \mathbf{x} \in \mathcal{B}_{\mathcal{X}}(\mathbf{x}_{o},\delta_{\mathbf{x}_{o}}))$ $\|\mathbf{f}(\mathbf{x})-\mathbf{f}(\mathbf{x}_{o})\|<\epsilon$. The following concept of \textit{equicontinuity} extends the concept of continuity to a collection/set $\mathcal{F}\subset\mathcal{C}(\mathcal{X},\mathcal{Y})$ of functions. 
\begin{definition}[\textit{Equicontinuity of a Set}] \cite[Chapter~7]{Munkres2000}\label{def:equicontinuity}
A function set $\mathcal{F}\subset\mathcal{C}(\mathcal{X},\mathcal{Y})$ is called equicontinuous at $\mathbf{x}_{o} \in \mathcal{X}$ if for every $\epsilon>0$ there exists $\delta_{\mathbf{x}_{o}}>0$ such that $(\forall \mathbf{x} \in \mathcal{B}_{\mathcal{X}}(\mathbf{x}_{o},\delta_{\mathbf{x}_{o}}))$ $(\forall \mathbf{f} \in \mathcal{F})$ $\|\mathbf{f}(\mathbf{x})-\mathbf{f}(\mathbf{x}_{o})\|<\epsilon$. Furthermore, if for every $\epsilon>0$ there exists $\delta>0$ such that $(\forall \mathbf{x}_{o} \in \mathcal{X})$ $(\forall \mathbf{x} \in \mathcal{B}_{\mathcal{X}}(\mathbf{x}_{o},\delta))$ $(\forall \mathbf{f} \in \mathcal{F})$ $\|\mathbf{f}(\mathbf{x})-\mathbf{f}(\mathbf{x}_{o})\|<\epsilon$, then $\mathcal{F}$ is said to be (uniformly) equicontinuous. 
\end{definition} 
\begin{remark}[Set of $\mathbf{L}$-Lipschitz Functions] \label{rem:one}
An example of a (uniformly) equicontinuous subset of $C(\mathcal{X},\mathcal{Y})$ is the set of $\mathbf{L}$-Lipschitz functions, i.e., Lipschitz functions with the Lipschitz constant determined by $\mathbf{L} \in \mathbb{R}_{\geq 0}^{M}$ (see Definition \ref{def:lf}). For completeness, a proof is shown in Appendix \ref{sec:app_zero}. 
\end{remark}
The general concept of compactness in normed vector spaces has been introduced in Definition \ref{def:boundedness_closedness_compactness}. The following Fact, along with Remark \ref{rem:compactness}, characterizes compact subsets of $C(\mathcal{X},\mathcal{Y})$. 
\begin{fact}[Compact subsets of $C(\mathcal{X},\mathcal{Y})$]\label{fact:arzela_ascoli}\cite{Brown1993}\cite[Corollary~45.5]{Munkres2000}
Let $\mathcal{X}$ be compact. Then,
\begin{itemize}
\item[a).] \textit{Arzel\'{a}-Ascoli's Theorem}: Every bounded and equicontinuous sequence $(\mathbf{f}_n)_{n \in \mathbb{N}} \subset C(\mathcal{X},\mathcal{Y})$ has a convergent subsequence.
\item[b).] A set $\mathcal{F} \subset C(\mathcal{X},\mathcal{Y})$ is compact if it is bounded, equicontinuous, and closed. 
\end{itemize}
\end{fact}
\begin{remark}[Compactness in $\mathbb{R}^{m}$ and in $C(\mathcal{X},\mathcal{Y})$]\label{rem:compactness}
A subset of a finite dimensional Euclidean space is compact \textit{if and only} if it is bounded and closed (see \textit{Heine-Borel Theorem} \cite[Theorem~27.3]{Munkres2000}). 
However, in $C(\mathcal{X},\mathcal{Y})$, equicontinuity is required in addition to boundedness and closedness for compactness. 
\end{remark}
Finally, we present the concept of \textit{implicit functions}, which plays an important role in our study. 
\begin{fact}[\textit{Implicit function theorem}]\cite{Krantz2003}\label{fact:implicit_function}
Consider sets $\mathcal{X}\subset\mathbb{R}^N$, $\mathcal{Y} \subset\mathbb{R} ^M$, and $\mathcal{Z}\subset\mathbb{R}^M$, and a vector-valued continuous function $\mathbf{g}:\mathcal{Y}\times\mathcal{X}\rightarrow \mathcal{Z}$. Denote by $(i\in\overline{1,M})$ $g_i:\mathcal{Y}\times\mathcal{X}\rightarrow \mathbb{R}$ the $i$th component of $\mathbf{g}$.
Now, assume that $\mathbf{g}$ is continuously differentiable in a neighborhood $(\exists \delta_{\overline{x}},\delta_{\overline{y}}>0)$ $\mathcal{B}_{\mathcal{Y}}(\overline{\mathbf{y}},\delta_{\overline{y}})\times \mathcal{B}_{\mathcal{X}}(\overline{\mathbf{x}},\delta_{\overline{x}})$ of a point $(\overline{\mathbf{y}},\overline{\mathbf{x}})\in \mathcal{Y}\times\mathcal{X}$, and that $\mathbf{g}(\overline{\mathbf{y}},\overline{\mathbf{x}})=\mathbf{0}$. Let the Jacobian of $\mathbf{g}$ with respect to variables $\mathbf{y}$ (i.e., the first argument), denoted by $\boldsymbol{\nabla}^{\mathbf{g}}_{\mathbf{y}}:\mathcal{Y}\times\mathcal{X}\to \mathbb{R}^{M \times M}$ and defined as
\[
  \boldsymbol{\nabla}^{\mathbf{g}}_{\mathbf{y}}:=
\left( {\begin{array}{*{25}	c}
\frac{\partial g_1}{\partial y_1}&\frac{\partial g_1}{\partial y_2}&\cdots &\frac{\partial g_1}{\partial y_M}\\
\vdots & \vdots & \ddots & \vdots\\
\frac{\partial g_M}{\partial y_1}&\frac{\partial g_M}{\partial y_2}&\cdots &\frac{\partial g_M}{\partial y_M}
 \end{array} } \right),
\]
be invertible at $(\overline{\mathbf{y}},\overline{\mathbf{x}})$. Then, there exists a (unique and continuous) ``implicit'' function $\mathbf{f}:\mathcal{B}_{\mathcal{X}}(\overline{\mathbf{x}},\delta_{\overline{x}}) \rightarrow \mathcal{B}_{\mathcal{Y}}(\overline{\mathbf{y}},\delta_{\overline{y}})$ such that $(\forall \mathbf{x} \in \mathcal{B}_{\mathcal{X}}(\overline{\mathbf{x}},\delta_{\overline{x}}))$ $\mathbf{g}(\mathbf{f}(\mathbf{x}),\mathbf{x})=\mathbf{0}$. Furthermore, $\mathbf{f}$ is continuously differentiable on $\mathcal{B}_{\mathcal{X}}(\overline{\mathbf{x}},\delta_{\overline{x}})$. The value of the Jacobian of $\mathbf{f}$ is given by
\begin{equation}
           (\forall \mathbf{x} \in \mathcal{B}_\mathcal{X}(\overline{\mathbf{x}},\delta_{\overline{x}}))~\boldsymbol{\nabla}^{\mathbf{f}}_{\mathbf{x}}(\mathbf{x})= - \big(\boldsymbol{\nabla}^{\mathbf{g}}_{\mathbf{y}}(\mathbf{f}(\mathbf{x}),\mathbf{x})\big)^{-1} \boldsymbol{\nabla}^{\mathbf{g}}_{\mathbf{x}}(\mathbf{f}(\mathbf{x}),\mathbf{x}),
\end{equation}
where $\boldsymbol{\nabla}^{\mathbf{g}}_{\mathbf{x}}:\mathcal{Y}\times\mathcal{X}\to \mathbb{R}^{M \times N}$ is the Jacobian of $\mathbf{g}$ with respect to variables $\mathbf{x}$ (i.e., the second argument) given by
\[
  \boldsymbol{\nabla}^{\mathbf{g}}_{\mathbf{x}}:=
\left({\begin{array}{*{25}	c}
\frac{\partial g_1}{\partial x_1}&\frac{\partial g_1}{\partial x_2}&\cdots &\frac{\partial g_1}{\partial x_N}\\
\vdots & \vdots & \ddots & \vdots\\
\frac{\partial g_M}{\partial x_1}&\frac{\partial g_M}{\partial x_2}&\cdots &\frac{\partial g_M}{\partial x_N}
 \end{array}}\right).
\]
\end{fact}
%

\section{System Model}\label{sec:system_model}
In this study we consider an urban cellular base station deployment consisting of $M \in \mathbb{N}$ base stations and $N \in \mathbb{N}$ users. We consider the downlink and we denote by $r_{j}\in\mathbb{R}_{>0}$ the rate of user $j \in \overline{1,N}$ per unit time. We collect the rates of all users in a vector $\mathbf{r}:=[r_1,r_2,\cdots,r_N]^{\intercal} \in \mathbb{R}^{N}_{>0}$.
\begin{table}
\caption{List of Variables}
\centering
  \begin{tabular}{| l | l |}
    \hline
		 \textbf{Description}                                        &  \textbf{Symbol}     \\
		 \hline
		 Number of base stations                                      &  $M$                  \\ 
		 Number of users                                              &  $N$                  \\ 
		 Set of base stations                                        &  $\mathcal{M}=\{1,2,\ldots,M\}$                  \\                  
		 Set of users                                                &  $\mathcal{N}=\{1,2,\ldots,N\}$                  \\
		 Set of users for base station $i$                           &  $\mathcal{N}(i)$                  \\
		 Rate of user $j$                                            &  $r_j \in \mathbb{R}_{>0}$                       \\
		 Minimum user rate vector                                           &  $r_\text{min} \in \mathbb{R}_{>0}$              \\
		 Device SNR between base station $i$ and user $j$            &  $\gamma_{ij}$    \\ 
		 Number of resource blocks                                   &  $R\in\mathbb{N}$                           \\
		 Bandwidth of each resource block                            &  $B\in\mathbb{R}_{>0}$                           \\
		 Cell-load		                                               &  $\boldsymbol{\rho} \in \mathbb{R}_{\geq 0}^M$   \\
		 Load mapping                                                &  $\mathbf{q}:\mathbb{R}_{\geq 0}^{M}\times\mathbb{R}_{>0}^{N}\rightarrow\mathbb{R}_{\geq 0}^{M}$                       \\
		 Base station transmit power                                 &  $\mathbf{p} \in \mathbb{R}^{M}_{>0}$            \\
		 Path-loss between base station $i$ and user $j$       &  $G_{i,j}\in\mathbb{R}_{>0}$            \\  
		 Space of continuous functions from $X$ to $\mathcal{Y}$     &  $C(\mathcal{X},\mathcal{Y})$                 \\
		 Lipschitz constant                                          &  $\mathbf{L} \in \mathbb{R}_{\geq 0}^{M}$		     \\
		 Euclidean open-ball centered at $\mathbf{x}\in\mathcal{X}$  &  $\mathcal{B}_{\mathcal{X}}(\mathbf{x},\delta)$                            \\
		 Network coherence time                                      &  $T_{\text{net}}\in\mathbb{R}_{>0}$                            \\
		 Sample acquisition time                                     &  $T_{\text{obv}}\in\mathbb{R}_{>0}$                            \\
		 Sample average time                                         &  $T_{\text{avg}}\in\mathbb{R}_{>0}$                            \\
		 Sample set size                                             &  $K \in \mathbb{N}$                            \\
		\hline
 \end{tabular}
\end{table}
%

\subsection{Load Coupling Model and the Feasible Rate Region} \label{load_coupling_model} 
    We now present the load-coupling model proposed in \cite{Siomina2012,Ho2014}, which has been shown to be sufficiently accurate in certain scenarios in practice \cite{Fehske2012, Mogensen2007, Shen2015T}. This model is based on the fact that time-frequency resources available at a base station are divided into physical resource blocks to facilitate resource allocation. The cell-load (at a base station) is defined to be the fraction of available resource blocks that are allocated to support the rates of the users associated with the base station. Resource blocks are allocated to users based on their rates and channel qualities given in terms of their average signal-to-interference-plus-noise ratios (SINRs). In the following we denote by $\mathcal{M}:=\{1,2,\ldots,M\}$ and $\mathcal{N}:=\{1,2,\ldots,N\}$ the set of base stations and users, respectively, and we denote by $\mathcal{N}(i)$ the set of users associated with base station $i\in\mathcal{M}$. 
			
   Consider the case where base station $i\in\mathcal{M}$ is serving user $j\in\mathcal{N}(i)$ and denote by $G_{i,j}$ the path-loss between base station $i$ and user $j$. 
The load-based SINR model represents the inter-cell interference from base station $k \in \mathcal{M}\setminus i$ as the product $p_{k}G_{k,j}\rho_{k}\geq 0$, where $p_{k}$ is the fixed transmit power of base station $k$ per resource block, and where $0<\rho_{k}\leq 1$ denotes the cell-load at base station $k$ \cite{Fehske2012}.
With this model in hand, the network layer (averaged) SINR of the wireless link between base station $i$ and user $j$ is expressed as \cite{Siomina2012,Ho2014} 
\begin{equation}
\gamma_{ij}(\boldsymbol{\rho})=\frac{p_{i}G_{i,j}}{\sum_{k \in \mathcal{M}\setminus i}p_{k}G_{k,j}\rho_{k}+\sigma^{2}},
\label{eqn:SINR}
\end{equation}
where $\boldsymbol{\rho}:=[\rho_1,\rho_2,...,\rho_M]^{\intercal} \in \mathbb{R}_{>0}$ is the vector of cell-load values at all base stations in the network and where $\sigma^{2}$ denotes noise power. Note that the denominator in \eqref{eqn:SINR} provides an interpretation of the cell-load as the probability of inter-cell interference from base station $k$ \cite{Siomina2012}. For further details of the model including its strengths and weaknesses see \cite{Siomina2012,Ho2014}. Let $R\in\mathbb{N}$ be the total number of resource blocks available at the base station, each with bandwidth $B\in\mathbb{R}_{>0}$. Given SINR $\gamma_{ij}(\boldsymbol{\rho})$, we assume that base station $i$ can reliably transmit at a rate $r_{ij}^s = B\log(1+\gamma_{ij}(\boldsymbol{\rho}))$ per resource block to user $j$. Thus, to ``support'' the rate $r_{j}$, base station $i$ has to allocate $\rho_{ij}=\frac{r_{j}}{r_{ij}^s}$ resource blocks to user $j$. Summing the resource block consumption over all $\mathcal{N}(i)$, we obtain the ``cell-load'' (in terms of total resource consumption) of base station $i \in \overline{1,M}$ 
\begin{equation}
\rho_{i}=\frac{1}{RB}\sum_{j\in\mathcal{N}(i)} \frac{r_{j}}{\log(1+\gamma_{ij}(\boldsymbol{\rho}))}.
\label{eqn:load_i}
\end{equation}  

Note that, we can express the right-hand side of \eqref{eqn:load_i} for the entire network as a vector-valued mapping 
\begin{align*}
\mathbf{q}:\mathbb{R}_{\geq 0}^{M}\times\mathbb{R}_{>0}^{N}\,\,\,\,\,\,\,\,&\to\,\,\,\,\,\,\,\,\,\mathbb{R}_{> 0}^{M}\\ 
(\boldsymbol{\rho},\mathbf{r})\,\,\,\,\,\,\,\,\,\,\,\,\,\,\,\,\,\,\,\,\,&\mapsto\,\,\,\,\,\,\,\,\left[ {\begin{array}{*{25}	c}
\frac{1}{RB}\sum_{j\in\mathcal{N}(1)} \frac{r_{j}}{\log(1+\gamma_{ij}(\boldsymbol{\rho}))}\\
\vdots\\
\frac{1}{RB}\sum_{j\in\mathcal{N}(M)} \frac{r_{j}}{\log(1+\gamma_{ij}(\boldsymbol{\rho}))}
 \end{array} } \right],
\end{align*}
which we refer to as the \textit{load mapping}. Given $\overline{\mathbf{r}}\in\mathbb{R}^{N}_{>0}$, it follows from \eqref{eqn:load_i} that the cell-load vector is the solution (if it exists) to the fixed point problem: \textit{Find $\boldsymbol{\rho}^{\ast}=[\rho_1^\ast,\rho_2^\ast,...,\rho^\ast_M]^{\intercal} \in \mathbb{R}_{\geq 0}^{M}$ such that:}
\begin{equation}
\boldsymbol{\rho}^{\ast}=\mathbf{q}(\boldsymbol{\rho}^{\ast},\overline{\mathbf{r}}). 
\label{eqn:load_mapping}
\end{equation} 

   
	 Since the cell-load is defined as a fraction of the available resources at the base station, a rate vector is \textit{feasible} (i.e., there are sufficient resource blocks available at all base stations to support rate of every user) if the solution (if it exists) to \eqref{eqn:load_mapping} satisfies $\boldsymbol{\rho}^{\ast}\leq \mathbf{1}$. For a given supported $\overline{\mathbf{r}}\in\mathbb{R}^{N}_{>0}$, the solution to \eqref{eqn:load_mapping} can be obtained by iterative fixed point algorithms as long as the network information (path-losses, powers, user association, etc. in \eqref{eqn:load_i}) required by these algorithms is available. In more detail, given $\mathbf{r} \in \mathbb{R}^N_{>0}$, the mapping $\Gamma_{\mathbf{r}}:\mathbb{R}_{\geq 0}^{M}\rightarrow\mathbb{R}_{>0}^{M}:\boldsymbol{\rho}\mapsto\mathbf{q}(\boldsymbol{\rho},\mathbf{r})$ is a \textit{positive concave mapping}, so it also belongs to the class of \textit{standard interference functions} \cite{Cavalcante2016,yates95}. Therefore, the following holds:
\begin{fact}[\textit{The unique fixed point solution}]\cite{yates95}\label{fact:fixed_point}
Suppose the rate vector $\overline{\mathbf{r}} \in \mathbb{R}_{>0}^{N}$ is feasible, then the solution set of \eqref{eqn:load_mapping} given by
\begin{equation*}
\mathrm{Fix}(\Gamma_{\overline{\mathbf{r}}}):=\left\{\boldsymbol{\rho}^{\ast}\in \mathbb{R}_{\geq 0}^{M}~|~\mathbf{0}<\Gamma_{\overline {\mathbf{r}}}(\boldsymbol{\rho}^{\ast})=\boldsymbol{\rho}^{\ast}\leq \mathbf{1}\right\}
\end{equation*}
contains at most one fixed point.
\end{fact}

   As mentioned previously in Section \ref{sec:contribution}, we incorporate prior knowledge about the cell-load in our learning framework presented in Section \ref{section:the_learning_problem} to ensure robust learning. To this end, Fact \ref{proposition:one} presents an important property of the cell-load, namely its monotonicity in the rate vector:
\begin{fact}\cite[Theorem 2]{Ho2014}
Consider any two feasible rate vectors $\mathbf{r}^{k},\mathbf{r}^{j} \in \mathcal{R}$ and the corresponding fixed points $\boldsymbol{\rho}^j\in\mathrm{Fix}(\Gamma_{\mathbf{r}^{j}}) \neq \emptyset$ and $\boldsymbol{\rho}^k\in\mathrm{Fix}(\Gamma_{\mathbf{r}^{k}})\neq \emptyset$. Then $\mathbf{r}^{j} \geq \mathbf{r}^{k}\implies\boldsymbol{\rho}^j\geq\boldsymbol{\rho}^k$.
\label{proposition:one}
\end{fact}

In the next section we define and study the \textit{feasible} rate region, which is the set of all rates supported by the network.

\section{Properties of the Feasible Rate Region}\label{sec:feasible_rate_demand_set}
 In light of Fact \ref{fact:fixed_point} and Fact \ref{proposition:one}, and given the minimum feasible rate vector $\mathbf{r}_{\text{min}}\in \mathbb{R}^{N}_{>0}$ (e.g., corresponding to the lowest order \textit{modulation and coding scheme} in the network) that induces the cell-load $\boldsymbol{\rho}_{\text{min}}\in\mathbb{R}^{M}_{>0}$, we are now in a position to define the feasible rate region and the set of cell-load vectors over this set.
\begin{definition}[\textit{Feasible Rate Region and the Cell Load Set}]\label{defn:feasible_region_fixed_points}
The feasible rate region is defined as 
\begin{equation}
\mathcal{R}:=\{\mathbf{r}\geq \mathbf{r}_{\text{min}}\in\mathbb{R}^{N}_{>0}~|~(\exists~\boldsymbol{\rho}^{\ast}\in\mathrm{Fix}(\Gamma_{\mathbf{r}}))~,\boldsymbol{\rho}_{\text{min}}\leq \boldsymbol{\rho}^{\ast}\leq\mathbf{1}\}. 
\end{equation}
Similarly, the feasible cell-load set is given by the set of fixed points (see Fact \ref{fact:fixed_point}) 
\begin{equation}
\mathcal{L}:=\left\{\boldsymbol{\rho}\in \mathbb{R}^{M}_{>0} ~|~(\exists~\mathbf{r}^{\ast}\in\mathcal{R})~,\boldsymbol{\rho}_{\text{min}}\leq \Gamma_{\mathbf{r}^{\ast}}(\boldsymbol{\rho})=\boldsymbol{\rho}\leq\mathbf{1}\right\}.
\end{equation}
\end{definition}

In the following we extend the prior knowledge in our learning framework by studying the feasible rate region $\mathcal{R}\in\mathbb{R}^{N}_{>0}$ in Definition \ref{defn:feasible_region_fixed_points}. In particular, we show in Theorem \ref{theorem:one} that $\mathcal{R}$ is compact. The compactness of $\mathcal{R}$ is also required for our results in Section \ref{section:the_learning_problem}. 
%
    
	Note that $\mathcal{R}$ is bounded from below by $\mathbf{r}_{\text{min}}\in\mathbb{R}^{N}_{>0}$. Since power, bandwidth, and the total number of resource blocks are fixed in \eqref{eqn:SINR} and \eqref{eqn:load_i}, and because the cell-load is monotonic in the user rate vector by Fact \ref{proposition:one}, arbitrarily large user rates cannot be supported. We state this fact formally in Lemma \ref{lemma:one}, which we use to prove compactness of $\mathcal{R}$ in Theorem \ref{theorem:one}. 
\begin{lemma}\label{lemma:lemma_one}
The feasible rate region is bounded.
\label{lemma:one}
\end{lemma}

We now present the main result of this section.
\begin{theorem}\label{theorem:one}
\textit{The feasible rate region is compact.}
\label{theorem:one}
\end{theorem}
\begin{proof}
Recall from Definition \ref{def:boundedness_closedness_compactness}(b) that a subset of a normed space is closed \textit{if and only if} it contains all of its limit points. We denote by $\text{clo}(\mathcal{R})$ the \textit{closure} of $\mathcal{R}$ in Definition \ref{defn:feasible_region_fixed_points}, which is the smallest closed set in $\mathbb{R}_{>0}^{N}$ containing $\mathcal{R}$. Similarly, denote by $\text{clo}(\mathcal{L})$ the closure of $\mathcal{L}$ in Definition \ref{defn:feasible_region_fixed_points}. Consider an arbitrary sequence $(\mathbf{r}_n,\boldsymbol{\rho}_n)_{n\in\mathbb{N}} \subset \mathcal{R}\times\mathcal{L}$, of tuples consisting of feasible rate vectors and the corresponding cell-load vectors. Suppose $(\mathbf{r}_n,\boldsymbol{\rho}_n) \to (\overline{\mathbf{r}},\overline{\boldsymbol{\rho}}) \in \text{clo}(\mathcal{R})~\times~\text{clo}(\mathcal{L})$. From \eqref{eqn:load_mapping} it follows that, given $\mathbf{r}_n$, $\boldsymbol{\rho}_n$ must be the solution to the fixed point problem with the load mapping $\mathbf{q}$. Therefore, we have 
\begin{equation}
(\forall n \in \mathbb{N})~\boldsymbol{\rho}_{\text{min}}\leq \boldsymbol{\rho}_n=\mathbf{q}(\boldsymbol{\rho}_n,\mathbf{r}_n)\leq \mathbf{1}.
\end{equation} 
Now, since $\mathbf{q}$ is continuous, we have
\begin{align*}
\boldsymbol{\rho}_{\text{min}}\leq\lim_{n \in \mathbb{N}}\boldsymbol{\rho}_n&=\lim_{n \in \mathbb{N}}\mathbf{q}(\boldsymbol{\rho}_n,\mathbf{r}_n) \leq \mathbf{1}\\
\boldsymbol{\rho}_{\text{min}}\leq\overline{\boldsymbol{\rho}}&=\mathbf{q}(\overline{\boldsymbol{\rho}},\overline{\mathbf{r}})\leq \mathbf{1}
\end{align*}
which implies that $(\overline{\mathbf{r}},\overline{\boldsymbol{\rho}})\in \mathcal{R}\times\mathcal{L}$. Thus, every convergent sequence in $\mathcal{R}$ has its limit in $\mathcal{R}$ which implies that $\mathcal{R}$ is closed. Now, according to Lemma \ref{lemma:one}, $\mathcal{R}$ is bounded and recall from Remark \ref{rem:compactness} that every bounded and closed subset of a finite dimensional Euclidean space is compact.  
\end{proof}

\section{Robust Learning of Cell-Load} \label{section:the_learning_problem}
 Building upon the results from the previous section we formulate the robust learning of cell-load. Note that the cell-load is modeled by the load-coupling model in \eqref{eqn:load_i}. This means that given the network information required by the model, we can calculate the value of the ``modeled'' cell load. However, as mentioned in Section \ref{sec:why_robust_cell_load_estimation}, dynamic wireless networks are in general difficult to model accurately. Therefore, in the following we present a framework to directly approximate the cell-load values in networks that may not follow the cell-load model accurately. We use the cell-load model in this study only to extract some useful prior knowledge. In addition to the monotonicity of the cell-load and the compactness of the feasible rate region $\mathcal{R}$ established in Theorem \ref{theorem:one}, we show in Theorem \ref{theorem:two} that the function that maps rates to cell-load is continuously differentiable and therefore Lipschitz continuous on $\mathcal{R}$. The Lipschitz continuity is then used to solve our robust optimization problem formulated in the following. 
	
   Let $\mathcal{D}=\{(\mathbf{r}^k,\boldsymbol{\rho}^k:=\mathbf{f}^{\ast}(\mathbf{r}^k)) \in \mathcal{R} \times \mathcal{L}, k \in \overline{1,K}\}$ be a sample set of rates and their corresponding cell-load values, where $\mathbf{f}^{\ast}:\mathcal{R} \rightarrow \mathcal{L}$ is assumed to be a continuous but unknown function, and where $\mathcal{R}$ and $\mathcal{L}$ are defined in Definition \ref{defn:feasible_region_fixed_points}. We denote by $C(\mathcal{R},\mathcal{L})$ the space of vector-valued continuous functions mapping $\mathcal{R}$ to $\mathcal{L}$, equipped with the norm defined in \eqref{eqn:cheb_norm}. Our objective is to learn a function $\mathbf{g}^{\ast}$ that approximates $\mathbf{f}^{\ast}(\mathbf{r})$ for any $\mathbf{r} \in \mathcal{R}$ which is a classical problem considered in, for example, \cite{Weinberger1959,Sukharev1992,traub1980}. As mentioned in Section \ref{sec:contribution} we are interested in a robust approximation of $\mathbf{f}^{\ast}$. To this end, we consider the minimax optimization problem that leads to robust solutions under uncertainties: 
\begin{problem}\label{problem:classical_minimax}\cite{Sukharev1992,Belford1972}		
Given $\mathcal{D}=\{(\mathbf{r}^k,\boldsymbol{\rho}^k) \in \mathcal{R} \times \mathcal{L}, k \in \overline{1,K}\}$, find $\mathbf{g}^{\ast}\in C(\mathcal{R},\mathbb{R}_{\geq 0}^{M})$ such that the worst-case error 
\begin{equation}
\text{E}_{w}(\mathbf{g})=\underset{\mathbf{f}\in C(\mathcal{R},\mathcal{L})}{\sup}\left\|\mathbf{f}-\mathbf{g}\right\|_{C(\mathcal{R})},
\label{eq:error}
\end{equation}
attains its minimum (if it exists) subject to: $(\forall k \in \overline{1,K})$ $\mathbf{g}(\mathbf{r}^{k})=\mathbf{f}(\mathbf{r}^{k})=\boldsymbol{\rho}^{k}$. 
\end{problem}

  It is known that Problem \ref{problem:classical_minimax} can be solved by restricting $\mathbf{f}^{\ast}$ to a compact subset of $C(\mathcal{R},\mathcal{L})$ and by computing finite tight upper and lower bounds on the values $(\forall \mathbf{r} \in \mathcal{R})$ $\mathbf{f}^{\ast}(\mathbf{r})$ \cite{Weinberger1959,Beliakov2006,Beliakov2005}.
If the only information available about $\mathbf{f}^{\ast}$ is that it satisfies the interpolation constraints in Problem \ref{problem:classical_minimax}, then computing tight bounds on unseen function values $\mathbf{f}^{\ast}(\mathbf{r})$ is not possible, no matter how large the sample set $\mathcal{D}$ is. However, if we impose an additional restriction on $\mathbf{f}^\ast$ that satisfies certain properties \cite{Weinberger1959}, then we can obtain tight bounds $\boldsymbol{\sigma}_\text{l}(\mathbf{r})$ and $\boldsymbol{\sigma}_\text{u}(\mathbf{r})$ such that $\boldsymbol{\sigma}_\text{l}(\mathbf{r})\leq \mathbf{f}^\ast(\mathbf{r}) \leq \boldsymbol{\sigma}_\text{u}(\mathbf{r})$, where $\boldsymbol{\sigma}_\text{l}(\mathbf{r})$ and $\boldsymbol{\sigma}_\text{u}(\mathbf{r})$ can be computed explicitly. The optimal approximation $\mathbf{g}^{\ast}(\mathbf{r})$ of $\mathbf{f}^{\ast}(\mathbf{r})$ is simply given by $\mathbf{g}^{\ast}(\mathbf{r})=\frac{\boldsymbol{\sigma}_\text{l}(\mathbf{r})+\boldsymbol{\sigma}_\text{u}(\mathbf{r})}{2}$ and the \textit{magnitude of uncertainty} $\frac{|\boldsymbol{\sigma}_\text{u}(\mathbf{r})-\boldsymbol{\sigma}_\text{l}(\mathbf{r})|}{2}$ is minimal \cite{Calliess2014}. Therefore, no matter how small the sample set $\mathcal{D}$ is we are guaranteed the minimum worst-case error \eqref{eq:error}. It is in this sense that we refer to the learning as being robust (see Section \ref{sec:why_robust_cell_load_estimation}). 
			
	 In \cite{Beliakov2006,Beliakov2005} the analysis is restricted to Lipschitz functions in which case the above mentioned additional restriction results from the Lipschitz continuity. Following this approach, and by considering the cell-load model, we show in Theorem \ref{theorem:two} that $\mathbf{f}^{\ast}$ belongs to the class of $\mathbf{L}$-Lipschitz-Monotone Functions (LIMF) (see Definition \ref{def:lmf}). Moreover, Proposition \ref{proposition:proposition_two} shows that this class is a compact subset of $C(\mathcal{R},\mathcal{L})$. The computation of the bounds $\boldsymbol{\sigma}_\text{1}(\mathbf{r})$ and $\boldsymbol{\sigma}_\text{u}(\mathbf{r})$ is presented in Fact \ref{fact:main}.
			
   In the following we denote by $\widetilde{\mathcal{R}} \subset \mathbb{R}^{N}_{>0}$ the set of all rate vectors (not necessarily feasible/supported) for which there exists a fixed point solution of \eqref{eqn:load_mapping}, i.e., $\widetilde{\mathcal{R}}:=\{\overline{\mathbf{r}}\in\mathbb{R}^{N}_{>0}~|(\exists~\overline{\boldsymbol{\rho}}\in \mathbb{R}^{M}_{>0})~\overline{\boldsymbol{\rho}}=\mathbf{q}(\overline{\boldsymbol{\rho}},\overline{\mathbf{r}})\}$. So we have $\mathcal{R} \subset \widetilde{\mathcal{R}}$.
\begin{theorem}
Consider the load mapping $\mathbf{q}:\mathbb{R}^{M}_{\geq 0}\times\mathbb{R}^{N}_{>0}\to\mathbb{R}^{M}_{>0}$ in \eqref{eqn:load_mapping}.
\begin{itemize}
\item[a).] There exists a continuously differentiable function $\mathbf{f}^{\text{imp}}:\widetilde{\mathcal{R}} \rightarrow \mathbb{R}^{M}_{>0}$ such that $(\forall \overline{\mathbf{r}} \in \widetilde{\mathcal{R}})$ $\mathbf{f}^{\text{imp}}(\overline{\mathbf{r}})=\overline{\boldsymbol{\rho}}=\mathbf{q}(\overline{\boldsymbol{\rho}},\overline{\mathbf{r}})$.
\item[b).] The restriction of $\mathbf{f}^{\text{imp}}$ to the feasible rate region $\mathcal{R} \subset \widetilde{\mathcal{R}}$ is a LIMF function. 
\end{itemize}
\label{theorem:two}
\end{theorem}
\begin{proof}
\begin{itemize}
\item[a).] From the uniqueness of the fixed point solution of \eqref{eqn:load_mapping} it follows that, for two solution pairs $(\overline{\boldsymbol{\rho}_1},\overline{\mathbf{r}}_1)$ and $(\overline{\boldsymbol{\rho}_2},\overline{\mathbf{r}}_2)$, if $\overline{\boldsymbol{\rho}}_1\neq\overline{\boldsymbol{\rho}}_2$, then we must have $\overline{\mathbf{r}}_1\neq\overline{\mathbf{r}}_2$. Thus, there exists a function $\mathbf{f}^{\text{imp}}:\widetilde{\mathcal{R}} \rightarrow \mathbb{R}^{M}_{>0}:\overline{\mathbf{r}} \mapsto \mathbf{f}^{\text{imp}}(\overline{\mathbf{r}})=\mathbf{q}(\mathbf{f}^{\text{imp}}(\overline{\mathbf{r}}),\overline{\mathbf{r}})$ that maps every feasible rate vector to a unique fixed point. We now show that $\mathbf{f}^{\text{imp}}$ is continuously differentiable on $\widetilde{\mathcal{R}}$. 

  Consider the function $\mathbf{g}: \mathbb{R}^{N}_{>0} \times \mathbb{R}^{M}_{>0} \rightarrow \mathbb{R}^{M}$ defined as $\mathbf{g}(\mathbf{r},\boldsymbol{\rho}):=\boldsymbol{\rho}- \mathbf{q}(\boldsymbol{\rho},\mathbf{r})$, where $\mathbf{q}$ is the load mapping in \eqref{eqn:load_mapping}, and note that $(\forall\overline{\mathbf{r}}\in \widetilde{\mathcal{R}})$ $(\overline{\boldsymbol{\rho}}=\mathbf{f}^{\text{imp}}(\overline{\mathbf{r}}))$ $\mathbf{g}(\overline{\mathbf{r}},\overline{\boldsymbol{\rho}})=\mathbf{0}$. We now show that $\mathbf{g}$ is continuously differentiable, and the Jacobian matrix $\boldsymbol{\nabla}^{\mathbf{g}}_{\boldsymbol{\rho}}(\overline{\mathbf{r}},\overline{\boldsymbol{\rho}})$ is non-singular (invertible), on $\widetilde{\mathcal{R}}\times\mathbb{R}^{M}_{>0}$ (see Fact \ref{fact:implicit_function}). 
	To show that $\mathbf{g}$ is continuously differentiable, we show that the Jacobians $\boldsymbol{\nabla}^{\mathbf{g}}_{\mathbf{r}}$ and $\boldsymbol{\nabla}^{\mathbf{g}}_{\boldsymbol{\rho}}$ are continuous. The two Jacobians are given in Appendix \ref{sec:app_one} and Appendix \ref{sec:app_two}, respectively, and it can be verified that they are continuous. The invertibility of the $M\times M$ matrix $\boldsymbol{\nabla}^{\mathbf{g}}_{\boldsymbol{\rho}}(\overline{\mathbf{r}},\overline{\boldsymbol{\rho}})$ is shown in Appendix \ref{sec:app_three}. 
Therefore, according to Fact \ref{fact:implicit_function}, $\mathbf{f}^{\text{imp}}$ is continuously differentiable.	
\item[b).] According to part (a) and Fact \ref{fact:implicit_function}, the Jacobian $\boldsymbol{\nabla}^{\mathbf{f}{\text{imp}}}_{\mathbf{r}}$ is continuous on $\widetilde{\mathcal{R}}$. Denote by $\mathbf{f}:\mathcal{R} \rightarrow \mathcal{L}$ and $\boldsymbol{\nabla}^{\mathbf{f}}_{\mathbf{r}}$, the restriction of $\mathbf{f}^{\text{imp}}$ and $\boldsymbol{\nabla}^{\mathbf{f}{\text{imp}}}_{\mathbf{r}}$, respectively, to the set of feasible rate vectors $\mathcal{R}\subset \widetilde{\mathcal{R}}$. Since $\mathcal{R}$ is compact according to Theorem \ref{theorem:one}, $\boldsymbol{\nabla}^{\mathbf{f}}_{\mathbf{r}}$ is bounded on $\mathcal{R}$ according to  the \textit{extreme value theorem} \cite{Munkres2000} which implies that $\exists \mathbf{L} \in \mathbb{R}_{\geq 0}^{M}$ such that $\mathbf{f}$ is $\mathbf{L}$-Lipschitz on $\mathcal{R}$. Moreover, by Fact \ref{proposition:one}, $\mathbf{f}$ is monotonic on $\mathcal{R}$, so $\mathbf{f}$ is a LIMF function (see Definition \ref{def:lmf}). 
\end{itemize} 
\end{proof}

   In the following we denote by $\mathcal{F} \subset C(\mathcal{R},\mathcal{L})$ the class of LIMF functions $\mathbf{f}:\mathcal{R} \rightarrow \mathcal{L}$ with a given $\mathbf{L} \in \mathbb{R}_{\geq 0}^{M}$ (see Definition~\ref{def:lmf}). Before we proceed further, we obtain the following important result whose proof is shown in Appendix~\ref{sec:proof_of_proposition_two}.
\begin{proposition}\label{proposition:proposition_two} 
The class $\mathcal{F} \subset C(\mathcal{R},\mathcal{L})$ of LIMF functions, with a given $\mathbf{L}=[L_1,L_2,\cdots,L_M]^{\intercal} \in \mathbb{R}_{\geq 0}^{M}$, is compact.
\end{proposition}

\subsection{Minimax Optimal Approximation}\label{sec:minimax_optimal_approximation}
 We are now in a position to incorporate the prior information obtained in previous sections into Problem \ref{problem:classical_minimax}. Moreover, we formally state the robust learning problem considered in this study as an optimization problem.
\begin{definition}[\textit{Minimax Optimal Approximation}]\label{def:optimal_approximation} 
\textit{Let $\mathcal{D}=\{(\mathbf{r}^k,\boldsymbol{\rho}^k) \in \mathcal{R} \times \mathcal{L}\}_{k=1}^{K}$ be a sample set and assume that $(\forall k\in\overline{1,K})$ $\boldsymbol{\rho}^k:=\mathbf{f}^{\ast}(\mathbf{x}^k)$ are values generated by an unknown function $(\mathcal{F}\ni)$ $\mathbf{f}^{\ast}:\mathcal{R}\rightarrow \mathcal{L}$, where $\mathcal{F}\subset C(\mathcal{R},\mathcal{L})$ is a set of LIMF functions with a given $\mathbf{L} \in \mathbb{R}_{\geq 0}^{M}$. The minimax optimal approximation problem can be then stated as follows:
\begin{problem}\label{problem:main}\cite{Sukharev1992,Beliakov2005,Belford1972}
Find $\mathbf{g}^{\ast}$ such that
\begin{equation}
\mathbf{g}^{\ast}\in\underset{\mathbf{g}\in S}{\arg\min}\,\text{E}_{\text{max}}(\mathbf{g})
\end{equation}
where $S:=\{\mathbf{g}\in C(\mathcal{R},\mathbb{R}_{>0}^{M})\,|\,(\forall k\in\overline{1,K})\,\mathbf{g}(\mathbf{r}^k)=\boldsymbol{\rho}^k\}$, and $\text{E}_{\text{max}}(\mathbf{g}):=\max_{\mathbf{f}\in \mathcal{F}}\|\mathbf{f}-\mathbf{g}\|_{C(\mathcal{R})}$ is the worst-case error \eqref{eq:error} computed over the set $\mathcal{F}$.
\end{problem}}
\end{definition}

   The study \cite{Beliakov2005} proposes a framework for interpolation of scalar Lipschitz functions defined over a compact set by using a \textit{central algorithm} \cite{Sukharev1992,traub1980}. This framework can be used to obtain a solution to Problem \ref{problem:main}. Furthermore, this method is also ``shape preserving'', i.e., the approximation preserves the Lipschitz continuity and monotonicity of the underlying original function. The following fact summarizes the important properties of an optimal solution obtained based on this framework. 
\begin{fact}\cite{Beliakov2005}
Let $\mathcal{D}=\{(\mathbf{r}^k,\boldsymbol{\rho}^{k}) \in \mathcal{R} \times \mathcal{L}\}_{k=1}^{K}$ be a dataset generated by an unknown function $\mathbf{f}^{\ast} \in \mathcal{F}$, where $\mathcal{F}$ is the set of LIMF functions with the same $\mathbf{L}:=[L_1,L_2,\cdots,L_M]^{\intercal} \in \mathbb{R}_{\geq 0}^{M}$. Then, the following holds: 
\begin{itemize}
\item[a).] A minimax optimal approximation $\mathbf{g}^{\ast}$ of $\mathbf{f}^{\ast}\in\mathcal{F}$ can be constructed component-wise by 
\begin{equation}
 (\forall i \in \overline{1,M})\,(\forall \mathbf{r} \in \mathcal{R}) \,\, g^{\ast}_{i}(\mathbf{r})=\frac{\sigma^{i}_{l}(\mathbf{r})+ \sigma^{i}_{u}(\mathbf{r})}{2},
\label{eqn:opt_inter}
\end{equation}
where $\sigma^{i}_\text{l}(\mathbf{r})=\max_k\{\rho_{i}^{k}-L_{i}\|(\mathbf{r}^{k}-\mathbf{r})_{+}\|\}$, $\sigma^{i}_\text{u}(\mathbf{r})=\min_k\{\rho_{i}^{k}+L_{i}\|(\mathbf{r}-\mathbf{r}^{k})_{+}\|\}$, and $L_{i} \in \mathbb{R}_{\geq 0}$ is the Lipschitz constant of the $i$th component $f^{\ast}_{i}$ of $\mathbf{f}^{\ast}$. 
\item[b).] \label{fact_two_three} The approximation preserves the $\mathbf{L}$-Lipschitz continuity and monotonicity, i.e., $\mathbf{g}^{\ast}$ is $\mathbf{L}$-Lipschitz and monotonic.
\item[c).] $\mathbf{g}^{\ast}$ interpolates the sample set $\mathcal{D}$. 
\end{itemize}
\label{fact:main}
\end{fact}
\subsection{Complexity}\label{sec:complexity}
The complexity of the closed-form computation \eqref{eqn:opt_inter} is linear in the sample size $K$, i.e., the complexity is $O(K)$. Since we consider very small sample sizes, the complexity is not of a practical concern. 
Moreover, \eqref{eqn:opt_inter} can be computed independently for each base station. Therefore, the complexity is independent of the number of base stations $M$.

\begin{remark}[Prior Knowledge Decreases Uncertainty] \label{rem:three}
Note that the study \cite{Beliakov2005} is concerned with shape preserving approximation and it does not consider learning from a small sample set. However, we show in Proposition \ref{proposition:proposition_three} that (except for one particular case) excluding prior information regarding monotonicity worsens at least one of the bounds in Fact \ref{fact:main}(a) during generalization on unseen data and this therefore increases uncertainty and error. We also evaluate this fact empirically in Section~\ref{sec:effect_of_prior_information} in a realistic wireless network. 
\end{remark}

   The lower and upper bounds without monotonicity constraints in Fact~\ref{fact:main} are given by $(i \in \overline{1,M})$ $\eta^{i}_\text{l}(\mathbf{r})=\max_k\{\rho_{i}^{k}-L_{i}\|\mathbf{r}^{k}-\mathbf{r}\|\}$ and  $\eta^{i}_\text{u}(\mathbf{r})=\min_k\{\rho_{i}^{k}+L_{i}\|\mathbf{r}-\mathbf{r}^{k}\|\}$. Let $\text{U}_{\text{mon}}(\mathbf{r}):=\frac{|\sigma^{i}_\text{u}(\mathbf{r})-\sigma^{i}_\text{l}(\mathbf{r})|}{2}$ denote the magnitude of uncertainty calculated from the bounds in Fact~\ref{fact:main}, and let $\text{U}(\mathbf{r}):=\frac{|\eta^{i}_\text{u}(\mathbf{r})-\eta^{i}_\text{l}(\mathbf{r})|}{2}$ denote the magnitude of uncertainty without monotonicity in the framework. 
\begin{proposition}
Let $\mathbf{r} \notin \mathcal{D}=\{(\mathbf{r}^k,\boldsymbol{\rho}^{k}) \in \mathcal{R} \times \mathcal{L}\}_{k=1}^{K}$, where $\mathcal{D}$ is the data set in Fact \ref{fact:main}. Then $\text{U}_{\text{mon}}(\mathbf{r}) \leq \text{U}(\mathbf{r})$ if 
\begin{itemize}
\item[a).] ($k^{\ast} \in argmax_k\{\rho_{i}^{k}-L_{i}\|(\mathbf{r}^{k}-\mathbf{r})\|\})$~$\mathbf{r}^{k^{\ast}}\geq \mathbf{r}$, and  
\item[b).] ($j^{\ast} \in argmin_j\{\rho_{i}^{j}+L_{i}\|(\mathbf{r}-\mathbf{r}^{j})\|\}$)~$\mathbf{r}^{j^{\ast}}\leq \mathbf{r}$;
 \end{itemize}
otherwise $\text{U}_{\text{mon}}(\mathbf{r}) < \text{U}(\mathbf{r})$.
\label{proposition:proposition_three}
\end{proposition}
\begin{proof}
Consider two vectors $\mathbf{x}, \mathbf{y} \in \mathbb{R}_{\geq 0}^{N}$ such that $\mathbf{x}\neq\mathbf{y}$. If $\mathbf{x} \geq \mathbf{y}$, then $\|(\mathbf{x}-\mathbf{y})_{+}\|=\|(\mathbf{x}-\mathbf{y})\|$ and $\|(\mathbf{y}-\mathbf{x})_{+}\| < \|(\mathbf{x}-\mathbf{y})\|$. Similarly, if $\mathbf{x} \leq \mathbf{y}$, then $\|(\mathbf{x}-\mathbf{y})_{+}\|<\|(\mathbf{x}-\mathbf{y})\|$ and $\|(\mathbf{y}-\mathbf{x})_{+}\| = \|(\mathbf{x}-\mathbf{y})\|$. If $\mathbf{x}$ and $\mathbf{y}$ are incomparable then $\|(\mathbf{y}-\mathbf{x})_{+}\| < \|(\mathbf{x}-\mathbf{y})\|$ and also $\|(\mathbf{x}-\mathbf{y})_{+}\| < \|(\mathbf{x}-\mathbf{y})\|$. 

   Now, if conditions a) and b) are satisfied simultaneously, then (by condition a)) for the lower bound we have 
	\begin{align*}
	\eta^{i}_\text{l}(\mathbf{r})&=\{\rho_{i}^{k^{\ast}}-L_{i}\|(\mathbf{r}^{k^{\ast}}-\mathbf{r})\|\}\nonumber\\
	                             &=\{\rho_{i}^{k^{\ast}}-L_{i}\|(\mathbf{r}^{k^{\ast}}-\mathbf{r})_{+}\|\} \nonumber\\
															 &\leq \max_k\{\rho_{i}^{k}-L_{i}\|(\mathbf{r}^{k}-\mathbf{r})_{+}\|\}=\sigma^{i}_\text{l}(\mathbf{r}). \nonumber
	\end{align*}
	Similarly, (by condition b)) $\sigma^{i}_\text{u}(\mathbf{r})\leq\eta^{i}_\text{u}(\mathbf{r})$. This proves the first claim of the proposition. Now suppose condition a) is violated, i.e., either $\mathbf{r}^{k^{\ast}}\leq \mathbf{r}$ or $\mathbf{r}^{k^{\ast}}$ and $\mathbf{r}$ are incomparable, then from the above discussion 
\begin{align*}
	\eta^{i}_\text{l}(\mathbf{r})&=\{\rho_{i}^{k^{\ast}}-L_{i}\|(\mathbf{r}^{k^{\ast}}-\mathbf{r})\|\}\nonumber\\
	                             &<\{\rho_{i}^{k^{\ast}}-L_{i}\|(\mathbf{r}^{k^{\ast}}-\mathbf{r})_{+}\|\} \nonumber\\
															 &\leq \max_k\{\rho_{i}^{k}-L_{i}\|(\mathbf{r}^{k}-\mathbf{r})_{+}\|\}=\sigma^{i}_\text{l}(\mathbf{r}). \nonumber
\end{align*}
Similarly, if condition b) is violated, $\sigma^{i}_\text{u}(\mathbf{r})<\eta^{i}_\text{u}(\mathbf{r})$ and the second claim follows.  
\end{proof}

 The consequence of Proposition \ref{proposition:proposition_three} is that $\text{U}_{\text{mon}}(\mathbf{r}) < \text{U}(\mathbf{r})$ whenever $\mathbf{r}$ violates either of the two conditions in Proposition \ref{proposition:proposition_three}. Therefore, including prior knowledge in our framework regarding monotonicity provably improves generalization/prediction on unseen data.

\section{Implementation in a Wireless Network}\label{sec:cell_load_approximation}
We have shown in Theorem~\ref{theorem:two} that there exists an implicit function $(\forall i \in \overline{1,M})$ $f_i:\mathcal{R}\rightarrow \,]0,1]$ mapping every $\mathbf{r} \in \mathcal{R}$ to a cell-load value $\rho_i$ at base station $i$. Furthermore, Fact~\ref{fact:main} shows that given a sample set $\mathcal{D}(i)=\{(\mathbf{r}^k,{f}_{i}(\mathbf{r}^k))\}_{k=1}^{K}$ at base station $i$ and the knowledge of the Lipschitz constant $L_i$, we can easily approximate the cell-load value ${f}_{i}(\mathbf{r})$ for $\mathbf{r} \notin \mathcal{D}(i)$. In this section we show how to implement our framework in an OFDMA-based wireless cellular network. To this end, we first look at how to calculate the cell-load, and then we show how to obtain an appropriate sample set at a base station.

\subsection{Cell-load Calculation}\label{sec:cell_load_calculation} 
In OFDMA-based networks, such as LTE networks, time is divided into fixed length slots known as \textit{subframes}.  
During a subframe, if a base station is active, it transmits to one or more users on a block of frequencies in its cell. Therefore, users are allocated subframes in time and bandwidth in frequency to match their rate requirements. A subframe together with its bandwidth is commonly referred to as a \textit{physical resource block}. To calculate the cell-load, we record the fraction of the total available physical resource blocks allocated by a base station on average during a total time period of $T_{\text{avg}}>0$, where $T_{\text{avg}}$ is a design parameter.

\subsection{Obtaining a Sample Set}\label{sec:sample_set}
We denote by $T_\text{net}>0$ the network coherence time during which the environment (network topology, channels, rate distribution, etc.) is assumed to be constant (see Section~\ref{sec:why_robust_cell_load_estimation}). Let $T_\text{obv}<T_\text{net}$ denote the sample observation time. We divide $T_\text{obv}$ in $K\in\mathbb{N}$ time windows of duration $T_\text{avg}$ each as shown in Figure~\ref{fig:figure_two}. To obtain a sample set $\mathcal{D}(i)=\{(\mathbf{r}^k,\rho_{i}^{k}={f}_{i}(\mathbf{r}^k))\}_{k=1}^{K}$ at each base station $i \in \overline{1,M}$, the cell-load values $\rho_{i}^{k}={f}_{i}(\mathbf{r}^k)$ can be calculated as in Section~\ref{sec:cell_load_calculation} for each time window $k \in \overline{1,K}$. The base stations can exchange the rate values of users associated with them with other base stations to obtain the rate vectors $\mathbf{r}^k$. 
\begin{figure}[h]
  \centering
 \includegraphics[width=0.4\textwidth]{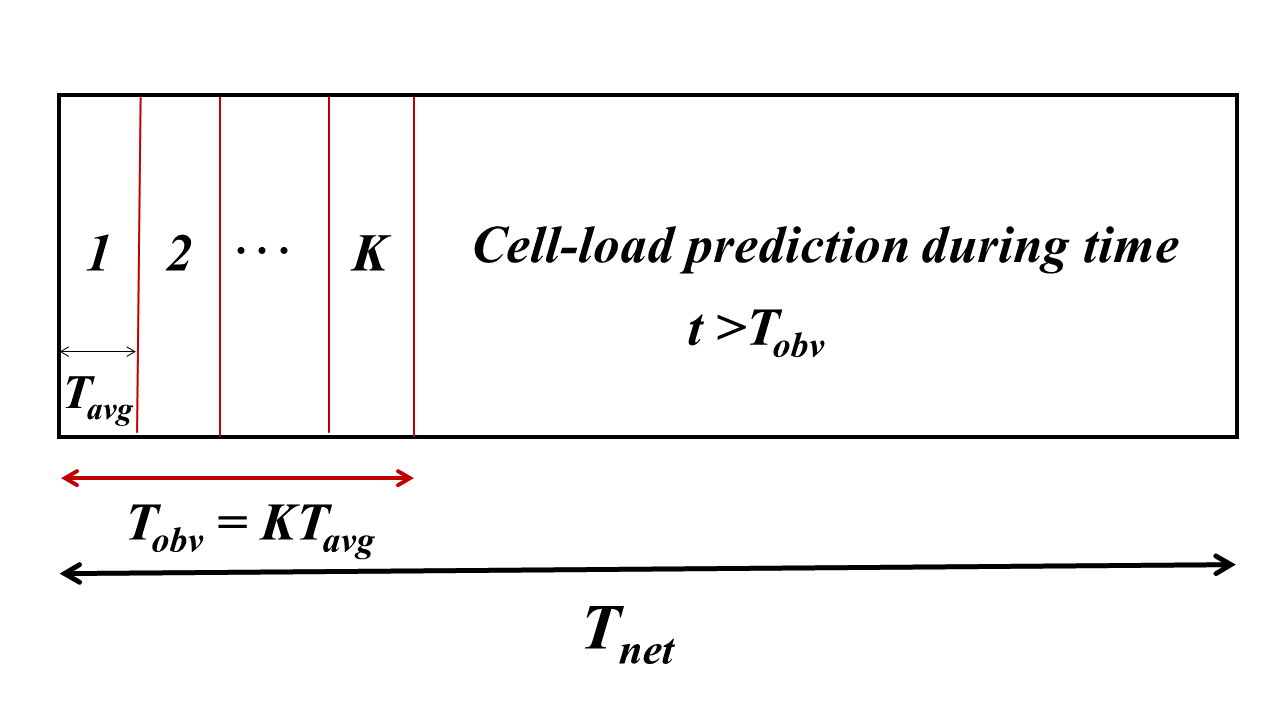}
 \caption{Learning Timeline: During each slot $k\in\overline{1,K}$ of length $T_{\text{avg}}$ we obtain a sample $(\mathbf{r}^k,\rho^{k}_i)$ by observing the proportion of resource blocks consumed to support rate $\mathbf{r}^k$ on average during $T_{\text{avg}}$.}
\label{fig:figure_two}
\end{figure}

  In the following, we assume that a sample set $\mathcal{D}(i)=\{(\mathbf{r}^k,\rho_{i}^{k}=f_{i}(\mathbf{r}^k))\}_{k=1}^{K}$, is available at time $t=T_\text{obv}$ at base station $i \in \overline{1,M}$. We also omit the index $i$ since the same procedure is carried out at each base station.  
\begin{algorithm}[t]
\small
\begin{algorithmic}
\caption{Cell-load Learning for each Base Station}
	\vspace{0.1cm}
 \State $\rightarrow$ \textbf{Initialization}  
 \begin{itemize} 
          \item Fix $K>0$ and $T_{\text{avg}}>0$.	
        \end{itemize}
		\vspace{0.1cm}		
\State $\rightarrow$ \textbf{Sample Acquisition}  (\textit{while} $t < T_\text{obv}$) 
				 \begin{itemize}
				   
          \item Exchange user rate with other base stations. 
       \end{itemize}
			\begin{itemize}
          \item Observe the sample set $\mathcal{D}^{\text{noise}}=\{(\mathbf{r}^k,y^{k}=f(\mathbf{r}^{k})+\epsilon(\mathbf{r}^k))\}_{k=1}^{K}$ (Section~\ref{sec:sample_set}).
       \end{itemize}
				\vspace{0.1cm}
\State $\rightarrow$ \textbf{Training} (\textit{at} $t = T_\text{obv}$) 
       \begin{itemize}
           \item Perform the estimation of $L$ (Section \ref{sec:smooth_measurements}).
					\end{itemize}
					\begin{itemize}
					\item  Perform data smoothing to obtain a compatible $\mathcal{D}^{\text{com}}=\{(\mathbf{r}^k,\tilde{\rho}^{k})\}_{k=1}^{K}$ (Section~\ref{sec:smooth_measurements}).  
       \end{itemize}
		\vspace{0.1cm}		
\State $\rightarrow$ \textbf{On-Demand Prediction} (\textit{at} $t > T_\text{obv}$) 
       \begin{itemize} 
          \item Given a new rate vector $\mathbf{r}\in\mathcal{R}$, perform the computation \eqref{eqn:opt_inter} in Fact~\ref{fact:main}
       \end{itemize}
       \begin{align}
			 g(\mathbf{r})\,=\,&\frac{1}{2}(\underset{k}{\max}\{\tilde{\rho}^{k}-L\|(\mathbf{r}^{k}-\mathbf{r})_{+} \| \})
			+\,\frac{1}{2}(\underset{k}{\min} \{\tilde{\rho}^{k}+L \|(\mathbf{r}-\mathbf{r}^{k})_{+}\| \}).\nonumber
          %
         \end{align}
					\label{algorithm:one}
\end{algorithmic}
\end{algorithm}

\subsection{Obtaining a Compatible Sample Set}\label{sec:smooth_measurements}
Note that the cell-load values calculated in a real network do not follow the cell-load model exactly.  
In more detail, instead of the sample set $\mathcal{D}=\{(\mathbf{r}^k,\rho^{k}=f(\mathbf{r}^{k}))\}_{k=1}^{K}$, we assume that an inaccurate sample set $\mathcal{D}^{\text{error}}=\{(\mathbf{r}^k,y^{k}=f(\mathbf{r}^{k})+\epsilon(\mathbf{r}^k))\}_{k=1}^{K}$ is available; $\epsilon(\mathbf{r}^k)\geq 0$ is the inaccuracy/error which is assumed to be bounded.\footnote{Our approximation framework is a special case of bounded error estimation/robust set-membership estimation \cite{Milanese1985,Milanese1991} which was developed for scenarios where the inaccuracy is unknown but bounded.} As a consequence, for a given value of the Lipschitz constant $L\in \mathbb{R}_{\geq 0}$, $\mathcal{D}^{\text{error}}$ may not be compatible with the monotonicity of $f$. Therefore, and if required, it must be smoothed to obtain a compatible set. Furthermore, in practice the prior information about the Lipschitz constant $L$ is often unavailable, so its value must be estimated from the set $\mathcal{D}^{\text{error}}$. In more detail, we first estimate the Lipschitz constant by $\tilde{L}:=\max_{k\neq j} \frac{|y^{k}-y^{j}|}{\|\mathbf{r}^{k}-\mathbf{r}^{j}\|}$ \cite{Strongin1973}.\footnote{There exist more sophisticated methods of estimating the Lipschitz constant such as the method proposed in \cite{Beliakov2005}. But these methods are not the focus of this study and they add substantial complexity to the algorithm.}
Given an estimate $\tilde{L}$ of the Lipschitz constant, we perform monotone-smoothing of $\mathcal{D}^{\text{error}}$. The details are provided in Appendix~\ref{sec:monotone_smoothing_problem}.

\subsection{Algorithm}\label{sec:algorithm}
The robust cell-load learning algorithm is presented in Algorithm \ref{algorithm:one}. The \textit{Sample Acquisition} step corresponds to the acquisition of the training sample set as explained in Section~\ref{sec:sample_set}, whereas \textit{Training} refers to Lipschitz constant estimation and the data smoothing process as presented in Appendix~\ref{sec:monotone_smoothing_problem}. The \textit{On-Demand Prediction} refers to the approximation of the cell-load value for a new rate vector during time period $T_{\text{net}}-T_{\text{obv}}$ (also see Figure~\ref{fig:figure_two}).

\section{Numerical Evaluation}\label{sec:numerical_evaluation}
In this section we evaluate the robust learning framework presented in Section \ref{sec:minimax_optimal_approximation} by simulation. To evaluate the learning techniques in a realistic cellular network, simulations are performed in the \textit{network simulator} (NS3) \cite{NS3}. 
We focus on the following aspects in this numerical evaluation:
\begin{enumerate}
\item We only use the load-coupling model (see Section \ref{load_coupling_model}) in this study to establish some prior knowledge about the cell-load in a real cellular network. We show in the simulations that our learning framework is able to predict the cell-load sufficiently accurately in a realistic cellular network in NS3. This is significant because models are only idealizations, and they may not capture the true behavior of cellular networks.    
\item We have shown in Proposition~\ref{proposition:proposition_three} that including prior knowledge decreases the uncertainty. We demonstrate this by comparing our learning framework with full prior knowledge with the case in which the prior information regarding the monotonicity of the cell-load with respect to rate is not included in the framework. 
\item Finally, we compare our method to standard multivariate regression techniques. 
We show the effect of sample size $K$ and the size of the network (i.e., the number of users $N$ and base stations $M$) on the quality of approximation.
\end{enumerate}

In the next section we present the LTE simulation framework in NS3.

\subsection{Network Simulator (NS3) and Scenario}\label{sec:network_simulator_and_scenario}
   We perform simulation in NS3 using the LTE model, the details of which can be found in \cite{NS3}. The load coupling model is evaluated in the LTE downlink in certain scenarios in \cite{Shen2015T}. Briefly, NS3 is a well-known discrete-event network simulator widely used in educational research and industry due to its accuracy in simulating computer networks such as LTE. The granularity of the LTE model in NS3 is up to the resource block level which allows for accurate packet scheduling and calculation of inter-cell interference. We chose the \textit{Round Robin} scheduler at the MAC layer. The reason is that the fairness inherent in the simple cyclic scheduling is more likely to ensure that the minimum data rate requirement of all users are met, which may not be the case with other more complex scheduling algorithms \cite{Dahlman2014}. The \textit{modulation and coding scheme} and the resource block allocation are chosen based on the wide-band \textit{channel quality indicator} (CQI). The CQI is calculated based on the average received SINR. 
   Users and base stations are distributed uniformly in the service area of $200\times 200$ meters. We perform simulations for $M=\{3, 5, 7, 9, 10\}$ base stations with $N=\{30,50,70,80, 90,100\}$ users. Users are associated with the base station to which they have the lowest path-loss. To generate training and test data, the data rates are distributed uniformaly between $0.1\times 10^{6}$ bits/s and $1\times 10^{6}$ bits/s. The important simulation parameters are shown in Table~\ref{table:two}. Other parameters were chosen as default in NS3. The simulation time was chosen to be $1$ second which is equal to the length $T_{\text{avg}}$ of each averaging time slot/window in Figure~\ref{fig:figure_two} and Algorithm~\ref{algorithm:one}. The cell-load values are calculated according to Section~\ref{sec:cell_load_calculation}.
\begin{table}
\caption{NS3 Simulation Parameters}
\centering
  \begin{tabular}{| l | l |}
    \hline
		 \textbf{Description}                   &  \textbf{Value}     \\
		 \hline
		 Number of base stations $M$            &  $3$                  \\                  
		 Number of users $N$                    &  $30$                  \\
		 Base station height                    &  $30$ m \\
		 User height                            &  $1.5$ m \\
     Noise figure base station              &  $5$ dB\\
		 Noise figure user                         &  $9$ dB \\
		 Min/Max user rate &  $0.1\times10^{6}/1\times10^{6}$ \\
		 Simulation area                        &  $200 \times 200$ m    \\ 
		 Simulation time                        &  $1$ s                           \\
		 Total bandwidth                        &  $10$ MHz                           \\
		 Total number of resource blocks        &  $50$  \\
		 Path-loss model                        &  Log-Distance Propagation Loss\\
		 SRS periodicity                        &  $80\times10^{-3}$ s \\
		 Internet application                   &  On-Off with Ipv4 \\ 
		\hline
 \end{tabular}
\label{table:two}
\end{table}

\subsection{Results}
We now present our numerical results. We use Algorithm $1$ to perform the robust learning of cell-load proposed in this study. We present the results for cell-load learning at a single base station. To obtain reliable statistics we consider $50$ topologies (with different user locations, base station locations, and user associations) for each value of $N$ and we let $M=N/10$. Note that scaling the number of base stations with an increase in the number of users is necessary to ensure that rate requirements of users are met. The objective of the simulation is to observe the effect of sample size and the network size on the approximation. 
For each fixed topology, we perform $100$ experiments for each value of $K \in \{10, 20,\ldots, 100\}$. During each experiment, a sample set $\mathcal{D}^{\text{error}}=\{(\mathbf{r}^k,y^k)\}_{k=1}^{K}$ is generated independently at random and the \textit{Training Step} is performed in Algorithm~\ref{algorithm:one} to obtain a compatible training sample set $\mathcal{D}^{\text{com}}$. Validation/prediction is performed for an independent test sample set of size $1000$ with rate vectors $\mathbf{r} \notin \mathcal{D}^{\text{com}}$. All results are averaged over $100$ experiments and then over $50$ topologies to obtain reliable statistics. 

\subsubsection{Effect of Prior Information}\label{sec:effect_of_prior_information}
In this section we compare our framework's performance with and without the prior information regarding the monotonicity of the cell-road with respect to rate (see Remark~\ref{rem:three}). For this simulation we consider $M=3$ and $N=30$. Note that the objective of this rather theoretical comparison is to confirm the result of Proposition~\ref{proposition:proposition_three} in a realistic simulation.  This comparison is performed with an ideal Lipschitz constant $L^{\text{ideal}}$ that can be obtained by using the method in Section~\ref{sec:smooth_measurements} but by using both the training sample set and the test sample set. This way $L^{\text{ideal}}$ is a good approximation of the true Lipschitz constant. We chose an ideal Lipschitz constant because in this section we want to focus only on the effect of including prior knowledge regarding monotonicity of the cell-load in rate in a realistic cellular network, and this requires an accurate calculation of function bounds in Section~\ref{section:the_learning_problem}. However, the comparison with \textit{state-of-art} techniques in Section~\ref{sec:comparison_with_state_of_art_techniques}, which is of a more practical significance, is performed with the Lipschitz constant that is estimated from only the training data set. 
\begin{figure}[h]
  \centering
 \includegraphics[width=0.5\textwidth]{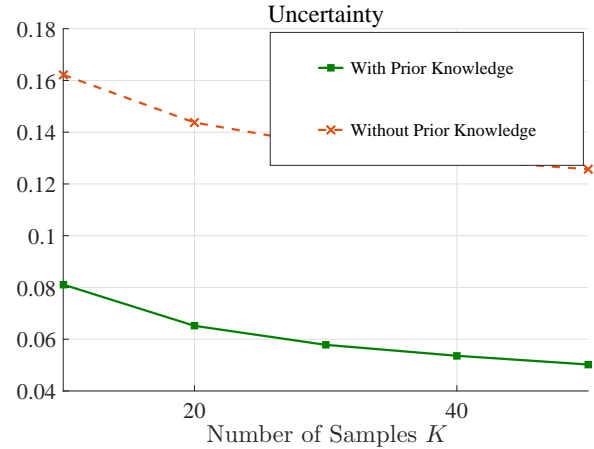}
 \caption{We compare the performance of our framework with the case where prior knowledge about the monotonicity of the cell-load has not been considered.}
\label{fig:figure_uncertainty}
\end{figure}
\begin{figure}[h]
  \centering
 \includegraphics[width=0.5\textwidth]{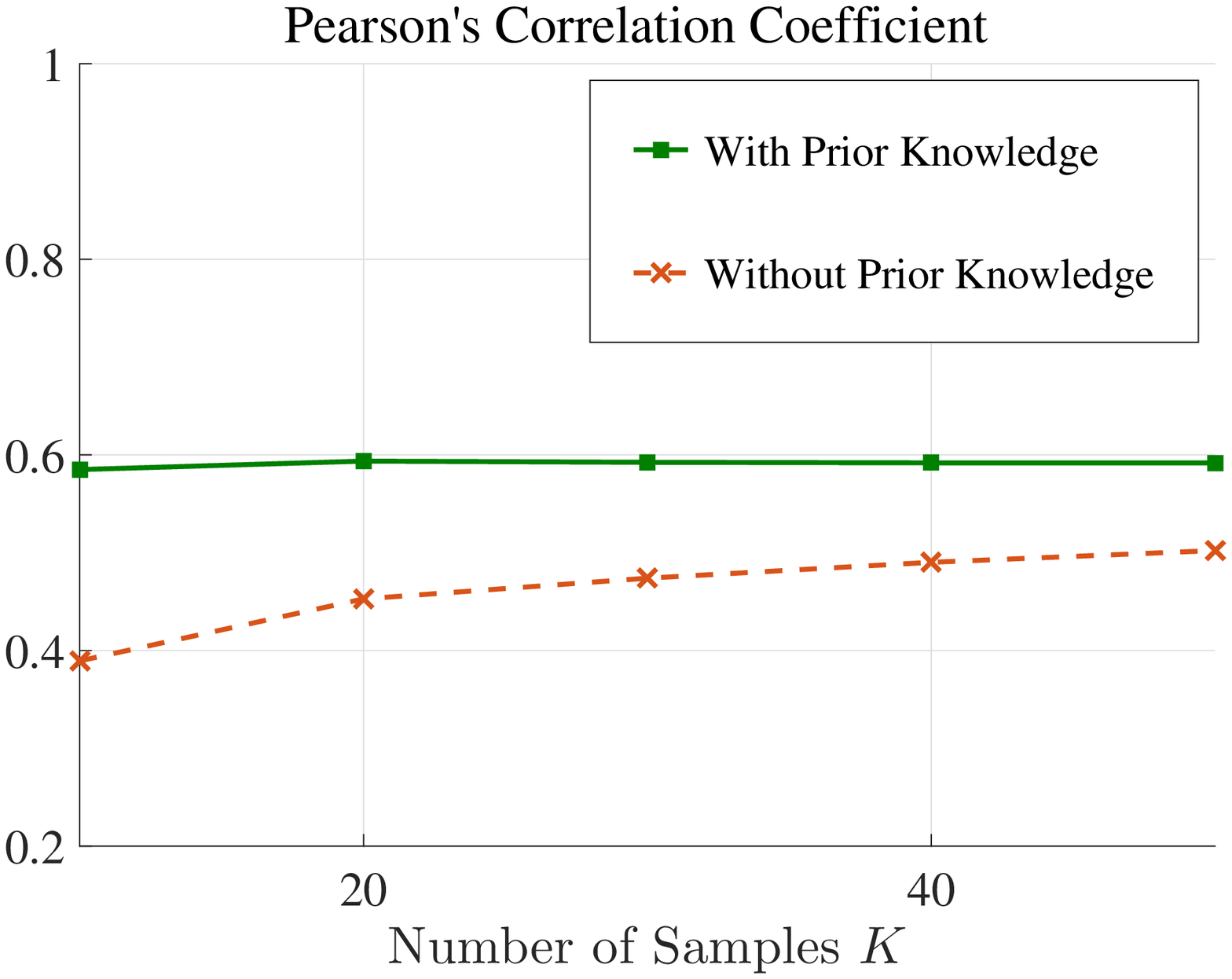}
 \caption{We compare the performance of LIMF learning framework with the case where prior knowledge about the monotonicity of the cell-load has not been considered.}
\label{fig:figure_corr_noprior}
\end{figure}
\begin{figure}[h]
  \centering
 \includegraphics[width=0.5\textwidth]{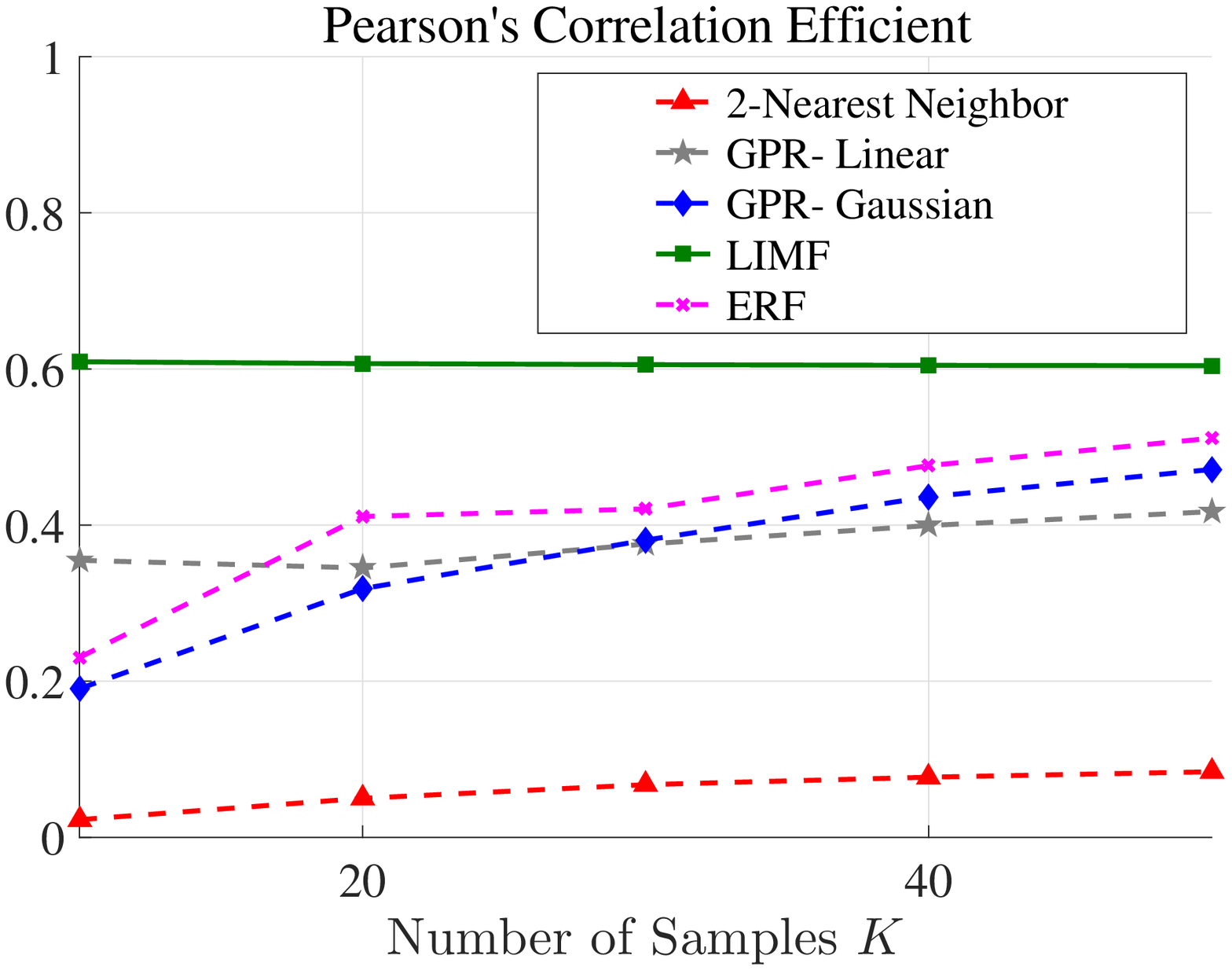}
 \caption{We compare the the $5$ techniques in terms of the linear correlation between predictions and true values for increasing $K$.}
\label{fig:figure_corr}
\end{figure}
\begin{figure}[h]
  \centering
 \includegraphics[width=0.5\textwidth]{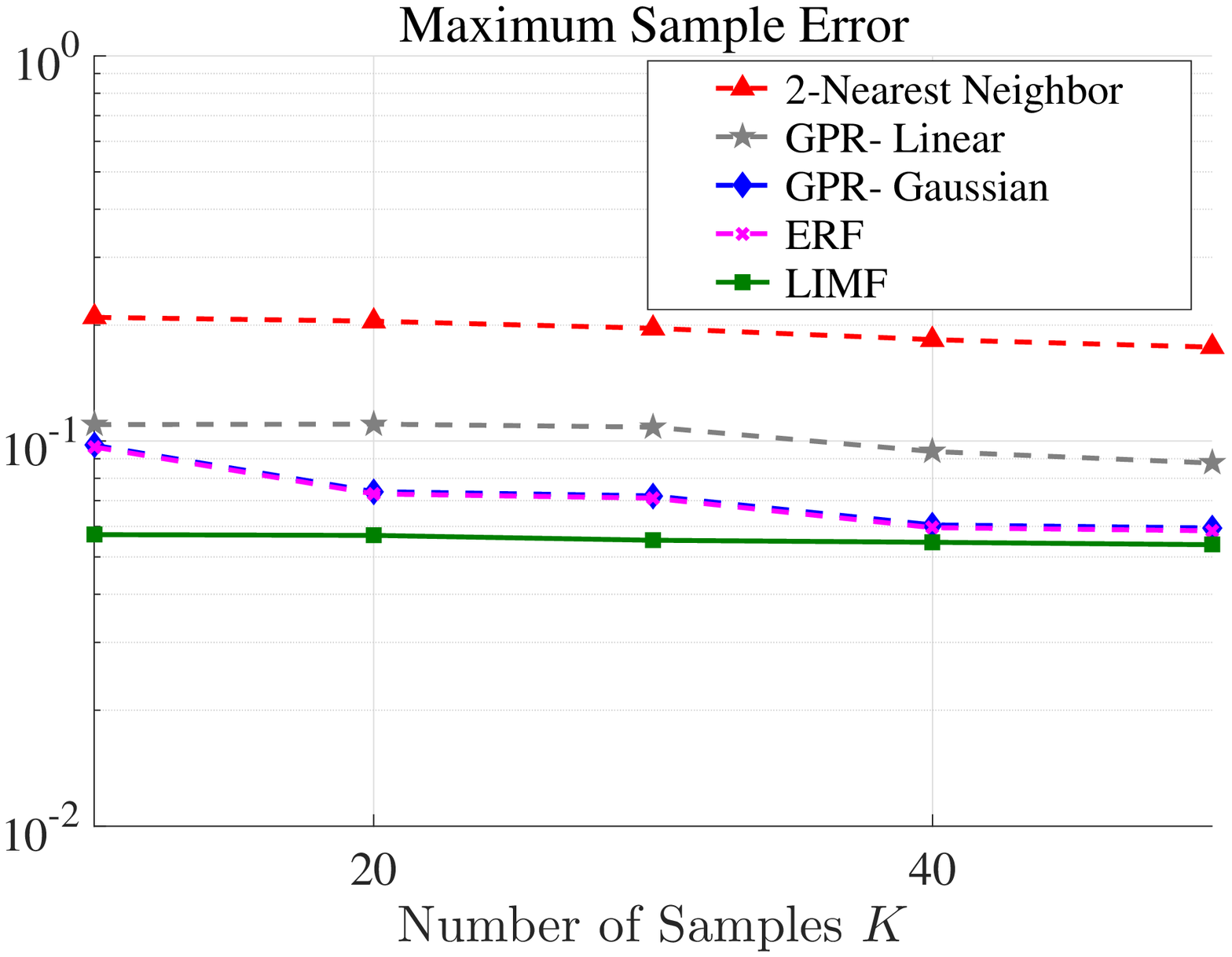}
 \caption{We compare the the $5$ techniques in terms of the maximum error between predictions and true values for increasing $K$.}
\label{fig:figure_quantile}
\end{figure}
\begin{figure}[h]
  \centering
 \includegraphics[width=0.5\textwidth]{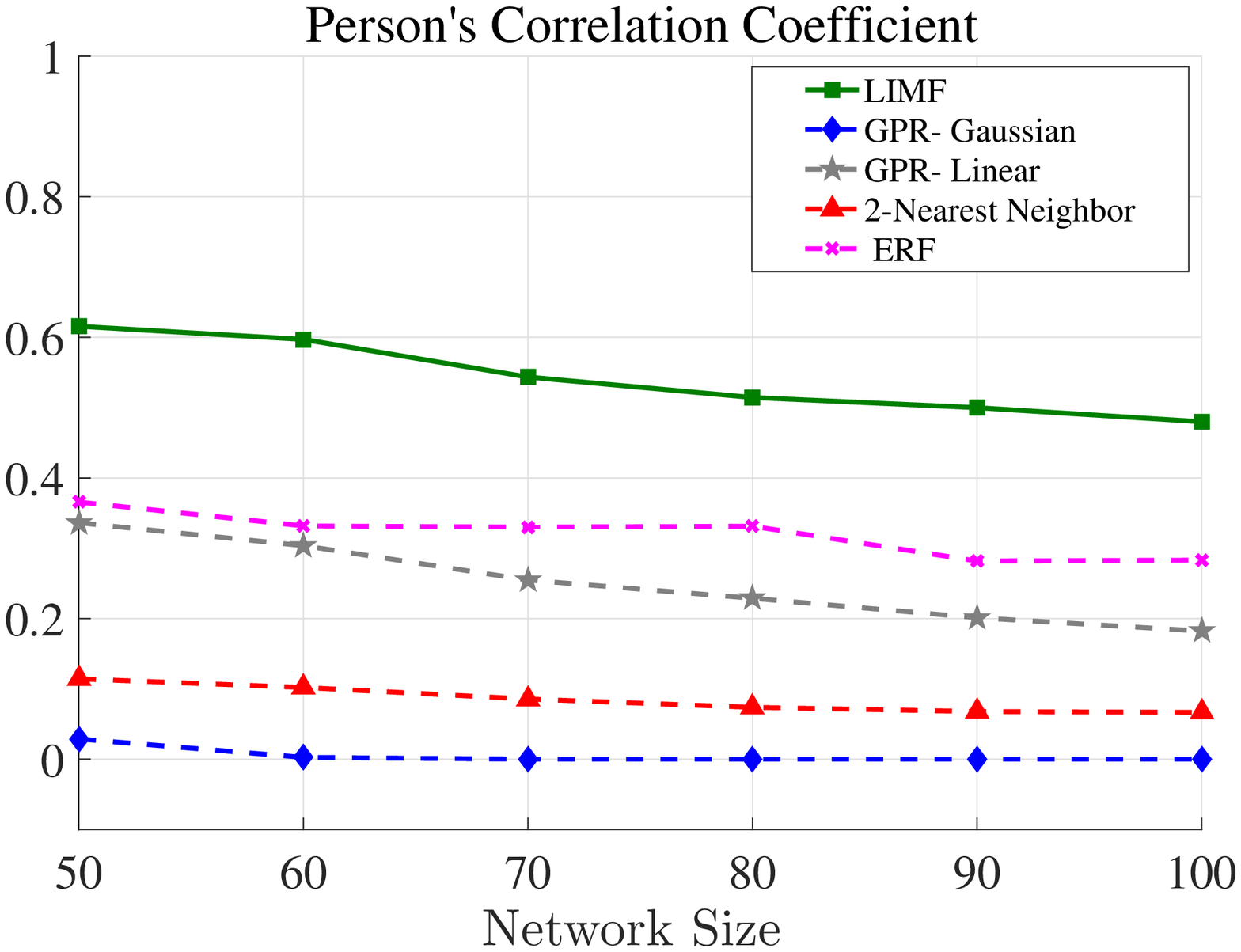}
 \caption{We compare the the $5$ techniques in terms of the linear correlation between predictions and true values for increasing network size.}
\label{fig:figure_corr_users}
\end{figure}
\begin{figure}[h]
  \centering
 \includegraphics[width=0.5\textwidth]{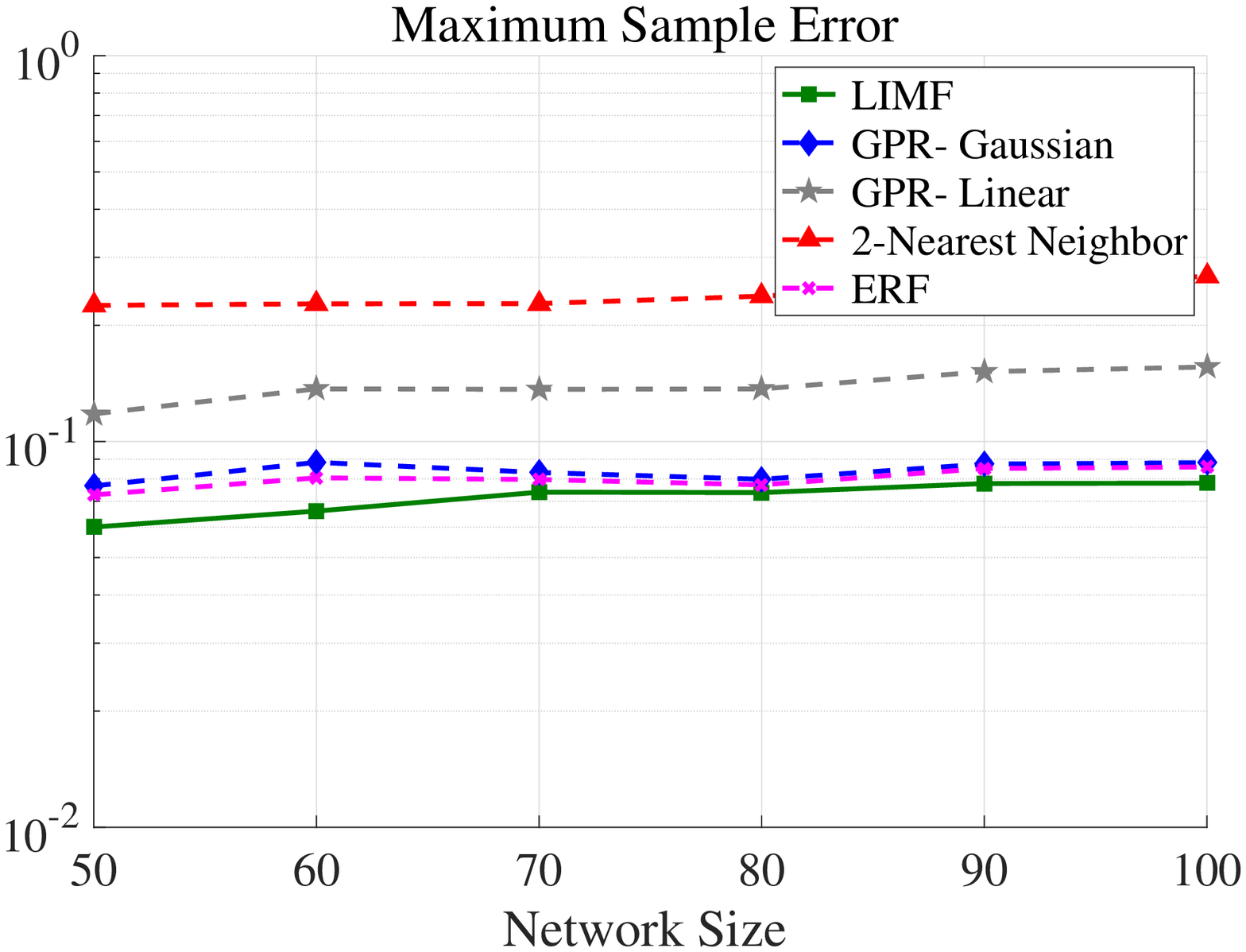}
 \caption{We compare the the $5$ techniques in terms of the maximum error between predictions and true values for increasing network size.}
\label{fig:figure_quantile_users}
\end{figure}

   We perform the comparison in terms of two metrics, namely the \textit{magnitude of uncertainty} given as $\frac{|\sigma_\text{u}(\mathbf{r})-\sigma_\text{l}(\mathbf{r})|}{2}$ (see Section~\ref{section:the_learning_problem}), where the rate $\mathbf{r}$ is a test sample point and $\sigma_\text{u}(\mathbf{r})$ and $\sigma_\text{l}(\mathbf{r})$ are upper and lower bounds, and the \textit{correlation} with test sample set that we measure in terms of the popular \textit{Pearson's correlation coefficient}. 

    The results are shown in Figure \ref{fig:figure_uncertainty} and Figure \ref{fig:figure_corr_noprior}. Figure~\ref{fig:figure_uncertainty} shows that uncertainty about the cell-load values decreases with the increasing training sample set size $K$ in both cases. However, we observe that the prior information regarding the monotonicity always results in less uncertainty than the case where monotonicity of the cell-load is ignored. 
		The results are therefore of a theoretical significance and they justify the inclusion of monotonicity as part of the prior knowledge in the framework (see Remark~\ref{rem:three}). The same effect is seen in Figure~\ref{fig:figure_corr_noprior} where we can clearly see that the case with all prior information included in the framework results in more correlation with the test sample set. 
%

\subsubsection{Comparison with State-of-Art Techniques}\label{sec:comparison_with_state_of_art_techniques}
In this section we compare our learning framework with some low-complexity state-of-art techniques for various training sample and network sizes. Throughout this section, we estimate $L$ from the available training sample set. We compare our method with four multivariate techniques, namely the state-of-art methods \textit{Gaussian process regression} (GPR) and \textit{ensemble learning with random forests} (ERF), and the simple \textit{2-nearest neighbor interpolation}. The GPR technique is well-known for its universal approximation of continuous functions defined over compact sets.
Note that, in addition to the state-of-art methods, it is important to compare the performance with a simple method such as the \textit{2-nearest neighbor interpolation} to highlight the difficulty of learning with small sample sets. We stress again that we consider very small sizes.  

  Figure~\ref{fig:figure_corr} shows a comparison of (linear) \textit{Pearson's} correlation coefficient, which is a popular measure of the strength and direction of the linear relationship between the predicted and the real test values, for an increasing sample size and fixed number of users $N=30$. In particular, we use this coefficient as a measure of the ``quality" of approximation. A high positive value of \textit{Pearson's} correlation coefficient means that the predictions made by the learning method have a strong linear relationship with the test sample set. 
Figure~\ref{fig:figure_quantile} shows the maximum or worst-case error encountered while predicting on the test sample set for an increasing sample size $K$ and fixed number of users $N=30$. The maximum error is more suitable for comparing the robustness of the approximation techniques than some other popular error metrics because it shows that all error residuals remain below this level. Therefore, the maximum error is a reasonable substitute for the maximum error of approximation in \eqref{eq:error} which we cannot compute directly. 
 
  It is important to analyze maximum error and correlation together to better understand the comparison between our learning framework and other techniques. We observe that even for an inexact value of Lipschitz constant $L$, our method outperforms other techniques. An interesting observation is the fact that the GPR method (with the Gaussian function) and ERF show a relatively good error performance in Figure~\ref{fig:figure_quantile} but a considerably smaller correlation in Figure~\ref{fig:figure_corr} than our method for small sampze sizes $K<30$. This is because of the fact that our method incorporates prior knowledge about the cell-load and other methods do not. The poorest performance is seen in the case of the \textit{2-nearest neighbor interpolation} whose performance improves slowly with increasing sample size. Clearly, this shows that we do not have enough samples to perform such a simple interpolation. 
		
		Finally, Figure~\ref{fig:figure_corr_users} and Figure~\ref{fig:figure_quantile_users} show the effect of network size (in terms of number of users $N$) on the performance of all techniques for a small sample size of $K=20$. We see that, as expected, there is a gradual degradation of performance for all techniques. In particular, we observe in Figure~\ref{fig:figure_corr_users} that the GPR with Gaussian function performs poorly due to insufficient training. 
		\begin{table}
\caption{Training Time Comparison on standard PC}
\centering
  \begin{tabular}{| l | l |}
    \hline
		 \textbf{Technique}                   &  \textbf{Average Training Time}     \\
		 \hline
         LIMF              & $10\times 10^{-3}$ seconds\\
				 Nearest Neighbor  & not applicable\\
				 GPR               & $80\times 10^{-3}$ seconds\\
				 ERF               & $60 \times 10^{-3}$ seconds\\
		\hline
 \end{tabular}
\label{table:three}
\end{table}
\section{Conclusion}
We have studied the problem of robust learning of cell-load in dynamic wireless cellular networks with small sample sets. In this challenging setting, we have proposed a learning framework that is robust against uncertainties that result from learning based on a small training sample set. We have shown that robustness can be achieved with the help of some prior knowledge about the cell-load and its relationship with downlink rates. For example, an inherent property of the cell-load is that it is monotonic in rates so this property can be used as prior knowledge. To obtain additional prior knowledge, we have shown that the feasible rate region is compact, and that there exists a Lipschitz continuous function mapping feasible rates to the cell-load. These properties enables us to use the classical framework of minimax approximation. In this framework the objective is to minimize the worst-case error given a training sample set and the prior knowledge. We have shown by simulations in NS3 that, in a realistic scenario, our method outperforms other popular learning techniques. An extension of this study is to develop sophisticated methods for estimation of the Lipschitz constant from small sample sets.

\appendix
\section{}
\subsection{Proof of Equicontinuity of $\mathbf{L}$-Lipschitz functions}\label{sec:app_zero}
Let $\mathcal{F}\subset C(\mathcal{X},\mathcal{Y})$ denote the set of $\mathbf{L}$-Lipschitz functions with $\mathbf{L}:=[L_1,L_2,\cdots,L_M]^{\intercal}\in\mathbb{R}^{M}_{\geq 0}$. Since each component of $\mathbf{f} \in \mathcal{F}$ is Lipschitz on $\mathcal{X} \subset \mathbb{R}_{>0}^N$, we have that
\begin{equation}
(\forall \mathbf{x},\mathbf{y} \in \mathcal{X})~(\forall i \in \overline{1,M})~\left|f_{i}(\mathbf{x})-f_{i}(\mathbf{y})\right|\leq L_{i}\left\|\mathbf{x}-\mathbf{y}\right\|.
\end{equation}

Define $L_{\text{max}}:=\max_{i \in \overline{1,M}} L_{i}$ and note that 
\begin{equation}
 (\forall \mathbf{x},\mathbf{y} \in \mathcal{X})\,\,\|\mathbf{f}(\mathbf{x})-\mathbf{f}(\mathbf{y})\|_{\infty} \leq L_{\text{max}} \left\|\mathbf{x}-\mathbf{y}\right\|. 
\end{equation} 
From the equivalence of norms in finite dimensional normed spaces it follows that $(\exists C > 0)$ such that 
\begin{equation}
\|\mathbf{f}(\mathbf{x})-\mathbf{f}(\mathbf{y})\| \leq C~\|\mathbf{f}(\mathbf{x})-\mathbf{f}(\mathbf{y})\|_{\infty} \leq C~L_{\text{max}} \left\|\mathbf{x}-\mathbf{y}\right\|. 
\label{eq:temp1}
\end{equation}

Given $\epsilon >0$ and for every $\mathbf{x}_{o} \in \mathcal{X}$, choose $\delta:=\frac{\epsilon}{L_{\text{max}}\,C}$ as the radius of $B_{{\mathcal{X}}}(\mathbf{x}_{o},\delta)$. We have from \eqref{eq:temp1} that 
\begin{equation}
\|\mathbf{f}(\mathbf{x})-\mathbf{f}(\mathbf{x}_{o})\| \leq C \, L_{\text{max}} \left\|\mathbf{x}-\mathbf{\mathbf{x}_{o}}\right\| < \epsilon, 
\label{eq:temp2}
\end{equation}
whenever $\left\|\mathbf{x}-\mathbf{\mathbf{x}_{o}}\right\|<\delta$. We have shown that $\delta$ can be chosen independently of $\mathbf{x}_{o}$. Now since \eqref{eq:temp2} holds for every $\mathbf{f} \in \mathcal{F}$, the proof is complete. 
\subsection{Jacobian of $\mathbf{g}$ with respect to $\mathbf{r}$}\label{sec:app_one}
The entry $[\boldsymbol{\nabla}^{\mathbf{g}}_{\mathbf{r}}(\mathbf{r},\boldsymbol{\rho})]_{i,j}$ of the $M\times N$ Jacobian $\boldsymbol{\nabla}^{\mathbf{g}}_{\mathbf{r}}(\mathbf{r},\boldsymbol{\rho})$ is given by
 \[ 
     [\boldsymbol{\nabla}^{\mathbf{g}}_{\mathbf{r}}(\mathbf{r},\boldsymbol{\rho})]_{i,j}=\Bigg\{\begin{tabular}{cc}
               $-\frac{1}{RB \log(1+\gamma_{ij})}$, & if $j\in\mathcal{N}(i)$\\
               $0$, & otherwise 
                                                                   \end{tabular}
											                                                                                                        \]
where $\gamma_{ij}:=\frac{p_{i}G_{i,j}}{\sum_{k \in \mathcal{M}\backslash\left\{i\right\}}p_{k}G_{k,j}\rho_{k}+\sigma^{2}}$.
\subsection{Jacobian of $\mathbf{g}$ with respect to $\boldsymbol{\rho}$}\label{sec:app_two}
The entry $[\boldsymbol{\nabla}^{\mathbf{g}}_{\boldsymbol{\rho}}(\mathbf{r},\boldsymbol{\rho})]_{i,k}$ of the $M\times M$ Jacobian $\boldsymbol{\nabla}^{\mathbf{g}}_{\boldsymbol{\rho}}(\mathbf{r},\boldsymbol{\rho})$ is given by 
 \[ 
     [\boldsymbol{\nabla}^{\mathbf{g}}_{\boldsymbol{\rho}}(\mathbf{r},\boldsymbol{\rho})]_{i,k}=\Bigg\{\begin{tabular}{cc}
           $-\underset{j\in \mathcal{N}(i)}{\sum}\ln(2)\frac{r_j}{RB}\frac{\frac{p_{i}G_{i,j}}{p_{k}G_{k,j}}}{\ln^{2}(1+\gamma_{i,j})(\gamma_{i,j}^{-2}+\gamma_{i,j}^{-1})}$, & if $i \neq k$\\
           $1$, & if $i=k$ 

                                                                          \end{tabular}
											                                                                                                        \]
where $\gamma_{ij}:=\frac{p_{i}G_{i,j}}{\sum_{k \in \mathcal{M}\backslash\left\{i\right\}}p_{k}G_{k,j}\rho_{k}+\sigma^{2}}$.
\subsection{Invertibility of the Jacobian $\boldsymbol{\nabla}^{\mathbf{g}}_{\boldsymbol{\rho}}(\mathbf{r},\boldsymbol{\rho})$}\label{sec:app_three}
We follow the analysis in \cite{Ren2014} which exploits the sufficient conditions for invertibility of a generalized diagonal dominant matrix \cite{Berman94} on the whole domain. 
In more detail, we show that the matrix $\boldsymbol{\nabla}^{\mathbf{g}}_{\boldsymbol{\rho}}(\mathbf{r},\boldsymbol{\rho})$ is invertible because it is an invertible generalized diagonal dominant matrix.
For any $\boldsymbol{\rho}\in\mathbb{R}_{>0}^{M}$
\begin{align}
[\boldsymbol{\nabla}^{\mathbf{g}}_{\boldsymbol{\rho}}(\mathbf{r},\boldsymbol{\rho})]_{i}\boldsymbol{\rho}=\rho_{i}-\sum_{j \in \mathcal{N}(i)}
\frac{r_j}{RB\log(1+\gamma_{i,j})}\times\nonumber\\
\frac{\frac{\sum_{k \in \mathcal{M}\backslash\left\{i\right\}} \rho_k p_k G_{k,j}}{p_i G_{i,j}}}{\ln(1+\gamma_{i,j})(\gamma_{i,j}^{-2}+\gamma_{i,j}^{-1})}\nonumber,
\end{align}
where $[\boldsymbol{\nabla}^{\mathbf{g}}_{\boldsymbol{\rho}}(\mathbf{r},\boldsymbol{\rho})]_{i}$ is the $i$th row of $\boldsymbol{\nabla}^{\mathbf{g}}_{\boldsymbol{\rho}}(\mathbf{r},\boldsymbol{\rho})$. Since $\frac{\sum_{k \in \mathcal{M}\backslash\left\{i\right\}}\rho_k p_k G_{k,j}}{p_i G_{i,j}}<\frac{\sum_{k \in \mathcal{M}\backslash\left\{i\right\}}\rho_k p_k G_{k,j}+\sigma^{2}}{p_i G_{i,j}}=\gamma_{i,j}^{-1}$ and $\ln(1+\gamma_{i,j})(\gamma_{i,j}^{-2}+\gamma_{i,j}^{-1})>\gamma_{i,j}^{-1}$ \cite{Ren2014}, we have 
\begin{align}
&\sum_{j \in \mathcal{N}(i)}
\frac{r_j}{RB\log(1+\gamma_{i,j})}\times\frac{\frac{\sum_{k \in \mathcal{M}\backslash\left\{i\right\}} \rho_k p_k G_{k,j}}{p_i G_{i,j}}}{\ln(1+\gamma_{i,j})(\gamma_{i,j}^{-2}+\gamma_{i,j}^{-1})}<\nonumber\\ 
& \sum_{j \in \mathcal{N}(i)}
\frac{r_j}{RB\log(1+\gamma_{i,j})}=\rho_{i}\nonumber
\end{align}
which implies that $[\boldsymbol{\nabla}^{\mathbf{g}}_{\boldsymbol{\rho}}(\mathbf{r},\boldsymbol{\rho})]_{i}\boldsymbol{\rho}>0$. Since the off-diagonal entries are all non-positive and diagonal entries are all non-negative, $\boldsymbol{\nabla}^{\mathbf{g}}_{\boldsymbol{\rho}}(\mathbf{r},\boldsymbol{\rho})$ satisfies the sufficient conditions for it to be an invertible generalized diagonal dominant matrix \cite{Ren2014,Berman94}.
%
\subsection{Proof of Proposition~\ref{proposition:proposition_two}}\label{sec:proof_of_proposition_two}
\begin{proof}
The class $\mathcal{F} \subset C(\mathcal{R},\mathcal{L})$ satisfies the following properties:
\begin{itemize}
\item[a).] \textit{Boundedness}: $\mathcal{F}$ is bounded because $(\forall \mathbf{f}\in \mathcal{F})$ $\left\|\mathbf{f}\right\|_{C(\mathcal{R})} \leq 1$.
\item[b).] \textit{Equicontinuity}: Since $\mathcal{F}$ is a set of $\mathbf{L}$-Lipschitz functions, $\mathcal{F}$ is an equicontinuous subset of $C(\mathcal{R},\mathcal{L})$ (see Remark \ref{rem:one}). 
\item[c).] \textit{Closedness}: The class $\mathcal{F}$ can be written as $\mathcal{F}=\mathcal{F}^{\text{Lip}} \ \bigcap \mathcal{F}^{\text{mon}}$, where $\mathcal{F}^{\text{Lip}}$ and $\mathcal{F}^{\text{mon}}$ are the sets of $\mathbf{L}$-Lipschitz functions and continuous monotone functions, respectively, in $C(\mathcal{R},\mathcal{L})$. Recall that the intersection of two closed sets is closed. Therefore, it is sufficient to show that $\mathcal{F}^{\text{Lip}}$ and $\mathcal{F}^{\text{mon}}$ are closed sets. For completeness, we show in Lemma \ref{lemma:theorem2} that $\mathcal{F}^{\text{mon}}$ and $\mathcal{F}^{\text{Lip}}$ are closed sets.
\end{itemize}
The proposition now follows from Fact \ref{fact:arzela_ascoli}. 
\end{proof}
\begin{lemma}\label{lemma:theorem2}
Consider the space $C(\mathcal{X},\mathcal{Y})$.
\begin{itemize}
\item[a).] The set of monotonic functions $\mathcal{F}^{\text{mon}}$ in $C(\mathcal{X},\mathcal{Y})$ is closed.
\item[b).] The set of $\mathbf{L}$-Lipschitz functions $\mathcal{F}^{\text{Lip}}$ in $C(\mathcal{X},\mathcal{Y})$ is closed.
\end{itemize}
\begin{proof}
\begin{itemize}
\item[a).] Let $(\mathbf{f}_n)_{n\in\mathbb{N}}\subset\mathcal{F}^{\text{mon}} \subset C(\mathcal{R},\mathcal{L})$ be an arbitrary convergent sequence of continuous monotone functions converging to some $\mathbf{g} \in C(\mathcal{R},\mathcal{L})$. Then from Definition \ref{def:mf}, and the fact that inequalities are preserved in the limit, it follows that:
\begin{align*}
(\forall \mathbf{x},\mathbf{y} \in \mathcal{R})~\mathbf{x} \leq \mathbf{y} \implies (\forall n\in\mathbb{N})~\mathbf{f}_n(\mathbf{x}) &\leq \mathbf{f}_n(\mathbf{y})\\
                                   (\forall \mathbf{x},\mathbf{y} \in \mathcal{R})~\mathbf{x} \leq \mathbf{y} \implies  \lim_{n \to \infty}{\mathbf{f}_n(\mathbf{x})} &\leq \lim_{n \to \infty}{\mathbf{f}_n(\mathbf{y})}\\
														                                                       (\forall \mathbf{x},\mathbf{y} \in \mathcal{R})~\mathbf{x} \leq \mathbf{y} \implies \mathbf{g}(\mathbf{x}) &\leq \mathbf{g}(\mathbf{y}),\\  
\end{align*}
which means that $\mathbf{g} \in \mathcal{F}^{\text{mon}}$. Since $(\mathbf{f}_n)_{n\in\mathbb{N}}$ was chosen arbitrarily, the above holds for every sequence in $\mathcal{F}^{\text{mon}}$ showing that $\mathcal{F}^{\text{mon}}$ is closed. 
\item[b).] Following the same idea as above, we show that the limit function $\mathbf{g} \in C(\mathcal{R},\mathcal{L})$ of an arbitrary sequence $(\mathbf{f}^{\text{Lip}}_n)_{n\in\mathbb{N}} \subset \mathcal{F}^{\text{Lip}} \subset C(\mathcal{R},\mathcal{L})$ is Lipschitz with the same $\mathbf{L}$, i.e., $\mathbf{g} \in \mathcal{F}^{\text{Lip}}$ also. Note that $\|\mathbf{f}^{\text{Lip}}_n - \mathbf{g}\|_{C(\mathcal{R})} \to 0$ if and only if $(\forall i \in \overline{1,M})$ $\|{f_i}^{\text{Lip}}_n  - {g_i}\|_{C(\mathcal{R})} \to 0$. Therefore, it suffices to show that $(i \in \overline{1,M})$ $g_i$, the limit of the sequence $({f_i}_n^{\text{Lip}})_{n \in \mathbb{N}}$, is Lipschitz with $L_i$, the $i$th component of $\mathbf{L}$.

   Now, since ${f_i}^{\text{Lip}}_n  \to {f_i}$ uniformly, for some $\epsilon>0$ there exists $N^{\epsilon}_1 \in \mathbb{N}$ such that $(\forall \mathbf{x} \in \mathcal{R})$ $|{f_i}(\mathbf{x}) - {f_i}^{\text{Lip}}_{N^{\epsilon}_1}(\mathbf{x})|<\epsilon$ which implies that there exists $N^{\epsilon} > N^{\epsilon}_1$ such that $(\forall \mathbf{x}\in \mathcal{R})$ $|{f_i}(\mathbf{x}) - {f_i}^{\text{Lip}}_{N^{\epsilon}}(\mathbf{x})|<\epsilon/2$. Then,
\begin{align*}
(\forall \mathbf{x} \in \mathcal{R})~(\forall \mathbf{y} \in \mathcal{R})~|{f_i}(\mathbf{x}) - {f_i}(\mathbf{y})|&=|{f_i}(\mathbf{x})+{f_i}^{\text{Lip}}_{N^{\epsilon}}(\mathbf{x})\nonumber\\
                                       & -{f_i}_{N^{\epsilon}}^{\text{Lip}}(\mathbf{x})+{f_i}^{\text{Lip}}_{N^{\epsilon}}(\mathbf{y})\nonumber \\
                                       & -{f_i}^{\text{Lip}}_{N^{\epsilon}}(\mathbf{y})-{f_i}(\mathbf{y})| \nonumber\\ 
                                       &  <\epsilon/2 + \epsilon/2 + L_i\|\mathbf{x}-\mathbf{y}\| \nonumber \\
																			 &  = \epsilon + L_i\|\mathbf{x}-\mathbf{y}\|. \nonumber \\
\end{align*}
Since the above holds for all $\epsilon>0$, it follows that 
\begin{equation}
(\forall \mathbf{x} \in \mathcal{R})~(\forall\mathbf{y} \in \mathcal{R})~|{f_i}(\mathbf{x}) - {f_i}(\mathbf{y})| \leq L_i\|\mathbf{x}-\mathbf{y}\|. \nonumber 
\end{equation} 
\end{itemize}
\end{proof}
 \end{lemma}

 \subsection{Monotone Smoothing of the Sample set}\label{sec:monotone_smoothing_problem}
We consider the monotone-smoothing problem which is formulated as a standard convex optimization problem. 
The author in \cite{Beliakov2005} has shown that a sample set $\mathcal{D}^{\text{com}}:=\{(\mathbf{r}^k,\tilde{\rho}^{k})\}_{k=1}^{K}$ is compatible with the monotonicity if and only if
it satisfies the following set of linear constraints \cite[Proposition 4.1]{Beliakov2005}
\begin{equation}
(\forall k \in \overline{1,K})~(\forall j \in \overline{1,K})~\tilde{\rho}^{k}-\tilde{\rho}^{j}\leq \tilde{L}\|(\mathbf{r}^{k}-\mathbf{r}^j)_{+}\|.
\label{eq:smoothing_constraints}
\end{equation}
Given the measured sample set $\mathcal{D}^{\text{error}}=\{(\mathbf{r}^k,y^k)\}_{k=1}^{K}$, we look for a compatible set $\mathcal{D}^{\text{com}}=\{(\mathbf{r}^k,\tilde{\rho}^{k})\}_{k=1}^{K}$ (that satisfies \eqref{eq:smoothing_constraints}) that is closest to $\mathcal{D}^{\text{error}}$ in the $\|\cdot\|_1$ sense. In more detail, let $\mathbf{y}=[y^1,\cdots,y^K]^{\intercal}$ and $\tilde{\boldsymbol{\rho}}=[\tilde{\rho}^1,\cdots,\tilde{\rho}^K]^{\intercal}$, then we minimize 
\begin{equation}
\|\mathbf{y}-\tilde{\boldsymbol{\rho}}\|_1=\sum_{k=1}^{K} |\tilde{\rho}^{k}-y^{k}|. 
\label{eq:residual}
\end{equation}
We now formalize this problem as a standard linear program (LP) which can be solved easily by any standard convex solver. Denote the $k$th residual in \eqref{eq:residual} by $q^{k}:=\tilde{\rho}^{k}-y^{k}$ and split $q^{k}$ into two parts $q_{+}^k$ and $q_{-}^k$ such that $q^k=q_{+}^k-q_{-}^k$. Substituting $(\forall l\in \overline{1,K})$ $q^{l}+y^{l}$ for $\tilde{\rho}^{l}$ into \eqref{eq:smoothing_constraints} and \eqref{eq:residual}, the monotone-smoothing problem can be written as an LP \cite{Beliakov2005}  
\begin{align}
\centering
     \underset{q_{+}^k,q_{-}^k \geq 0}{\text{minimize}}~&~\sum^{K}_{k=1}{|q^k|}        \nonumber\\
     \text{subject to}~&~(\forall k \in \overline{1,K})~(\forall j \in \overline{1,K}) \nonumber\\
     & q^k-q^j\leq y^{j}-y^{k} + \tilde{L}\|(\mathbf{r}^{k}-\mathbf{r}^{j})_{+}\|,                
		\label{eq:linear_constraints}
\end{align}
where $|q^k|=q_{+}^k+q_{-}^k$, and where $q_{+}^k,q_{-}^k \geq 0$ are the optimization variables. The smoothed compatible values follow from $\tilde{\rho}^{k}=y^{k}+q^k$. 

   Note that since we consider very small sample sizes $K$ and the constraint matrix, with rows given by \eqref{eq:linear_constraints}, is sparse, the above LP can be solved efficiently with standard convex solvers that exploit sparsity \cite{Liberti2016}. Therefore, the complexity of the smoothing step, which is performed only once after sample acquisition, is not of a practical concern. 

\section{}

\section*{Acknowledgment}
The work was supported by the German Federal Ministry of Education and Research under grant 16KIS0605. This work is also supported by the Federal Ministry of Education and Research of the Federal Republic of Germany (BMBF) in the framework of the project 5G NetMobil with funding number 16KIS0691. The authors alone are responsible for the content of the paper.

\ifCLASSOPTIONcaptionsoff
  \newpage
\fi



%

\bibliographystyle{IEEEtran}
\bibliography{library}

%




\end{document}